\documentclass{IEEEtran}

\usepackage{times,amssymb,amsmath,amsfonts,float,nicefrac,subfigure,float,accents,color,stmaryrd,hyperref,graphicx,bbm,mathrsfs,amsthm}
\usepackage[mathscr]{eucal}
\newcommand\remove[1]{}

\allowdisplaybreaks

\newtheorem{theorem}{Theorem}

\newtheorem{lemma}[theorem]{Lemma}
\newtheorem{proposition}[theorem]{Proposition}
\theoremstyle{remark}
\newtheorem{definition}{Definition}[section]
\theoremstyle{remark}
\newtheorem{remark}{Remark}[section]

\newcommand{\ba}{\begin{array}}
\newcommand{\ea}{\end{array}}
\newcommand{\be}{\begin{equation}}
\newcommand{\ee}{\end{equation}}
\newcommand{\bea}{\begin{eqnarray}}
\newcommand{\eea}{\end{eqnarray}}

\def\argmax{\text{argmax}}
\newcommand\nc\newcommand
\nc\bfa{{\boldsymbol a}}\nc\bfA{{\boldsymbol A}}\nc\cA{{\mathscr A}}
\nc\bfb{{\boldsymbol b}}\nc\bfB{{\boldsymbol B}}\nc\cB{{\mathscr B}}
\nc\bfc{{\boldsymbol c}}\nc\bfC{{\boldsymbol C}}\nc\cC{{\mathscr C}}
\nc\sC{{\mathscr C}}
\nc\bfd{{\boldsymbol d}}\nc\bfD{{\bfD}}
\nc\cD{{\mathscr D}}
\nc\bfe{{\boldsymbol e}}\nc\bfE{{\boldsymbol E}}\nc\cE{{\mathscr E}}
\nc\bff{{\boldsymbol f}}\nc\bfF{{\boldsymbol F}}\nc\cF{{\mathscr F}}
\nc\bfg{{\boldsymbol g}}\nc\bfG{{\boldsymbol G}}\nc\cG{{\mathscr G}}
\nc\bfh{{\boldsymbol h}}\nc\bfH{{\boldsymbol H}}\nc\cH{{\mathscr H}}
\nc\bfi{{\boldsymbol i}}\nc\bfI{{\boldsymbol I}}\nc\cI{{\mathscr I}}\nc\sI{{\mathscr I}}
\nc\bfj{{\boldsymbolj}}\nc\bfJ{{\boldsymbol J}}\nc\cJ{{\mathscr J}}
\nc\bfk{{\boldsymbolk}}\nc\bfK{{\boldsymbol K}}\nc\cK{{\mathscr K}}
\nc\bfl{{\boldsymboll}}\nc\bfL{{\boldsymbol L}}\nc\cL{{\mathscr L}}
\nc\bfm{{\boldsymbolm}}\nc\bfM{{\boldsymbol M}}\nc\cM{{\mathscr M}}
\nc\bfn{{\boldsymboln}}\nc\bfN{{\boldsymbol N}}\nc\cN{{\mathscr N}}
\nc\bfo{{\boldsymbolo}}\nc\bfO{{\boldsymbol O}}\nc\cO{{\mathscr O}}
\nc\bfp{{\boldsymbolp}}\nc\bfP{{\boldsymbol P}}\nc\cP{{\mathscr P}}
\nc\eP{{\EuScriptP}}\nc\fP{{\mathfrak P}}
\nc\bfq{{\boldsymbol q}}\nc\bfQ{{\boldsymbol Q}}\nc\cQ{{\mathscr Q}}
\nc\bfr{{\boldsymbol r}}\nc\bfR{{\boldsymbol R}}\nc\cR{{\mathscr R}}
\nc\bfs{{\boldsymbol s}}\nc\bfS{{\boldsymbol S}}\nc\cS{{\mathscr S}}
\nc\bft{{\boldsymbol t}}\nc\bfT{{\boldsymbol T}}\nc\cT{{\mathscr T}}
\nc\bfu{{\boldsymbol u}}\nc\bfU{{\boldsymbol U}}\nc\cU{{\mathscr U}}
\nc\bfv{{\boldsymbol v}}\nc\bfV{{\boldsymbol V}}\nc\cV{{\mathscr V}}
\nc\bfw{{\boldsymbol w}}\nc\bfW{{\boldsymbol W}}\nc\cW{{\mathscr W}}\nc\sW{{\mathscr W}}
\nc\bfx{{\boldsymbol x}}\nc\bfX{{\boldsymbol X}}\nc\cX{{\mathscr X}}
\nc\bfy{{\boldsymbol y}}\nc\bfY{{\boldsymbol Y}}\nc\cY{{\mathscr Y}}
\nc\bfz{{\boldsymbol z}}\nc\bfZ{{\boldsymbol Z}}\nc\cZ{{\mathscr Z}}

\newcommand{\bfit}{\bfseries\itshape}

\begin{document}
\title{Achieving Secrecy Capacity of the Wiretap Channel and Broadcast Channel with a Confidential Component}

\author{Talha~Cihad~Gulcu,~\IEEEmembership{Student Member,~IEEE,}
        Alexander~Barg,~\IEEEmembership{Fellow,~IEEE}
   
\thanks{T. C. Gulcu was with the Department
of ECE and Institute for Systems Research, University of Maryland, College Park, MD 20742, USA, e-mail: gulcu@umd.edu.
Research supported in part by NSF grant CCF1217245.}% <-this % stops a space
\thanks{A. Barg is with the Department of ECE and Institute for Systems Research, University
of Maryland, College Park, MD 20742, and IITP, Russian Academy of
Sciences, Moscow, Russia. Email: abarg@umd.edu. Research supported
in part by NSF grants CCF1217245, CCF1217894, and CCF1422955.}% <-this % stops a space
\thanks{Manuscript received October 2014; revised March 2015; accepted October 2016. The results of this paper were presented in part 
at the IEEE Information Theory Workshop, Jerusalem, Israel, April 27-May 1, 2015.}}

\markboth{IEEE Transactions on Information Theory, Vol. , No. ,  ~2016}%
{T. C. Gulcu and A.  Barg: Achieving Secrecy Capacity of the Wiretap Channel}

\maketitle

\begin{abstract} 
The wiretap channel model of Wyner is one of the first communication models with both reliability 
and security constraints. Capacity-achieving schemes for various models of the wiretap channel
have received considerable attention in recent literature. In this paper, we show that 
capacity of the general (not necessarily degraded or symmetric) wiretap channel under a ``strong secrecy constraint''
can be achieved using a transmission scheme based on polar codes. We also extend our construction to 
the case of broadcast channels with confidential messages defined by Csisz{\'a}r and K{\"orner}, achieving
the entire capacity region of this communication model.
\end{abstract}

\begin{IEEEkeywords}
Polar codes, chaining construction, strong secrecy, coordinate partition.
\end{IEEEkeywords}
 
 \IEEEpeerreviewmaketitle
 
\section{Introduction}

The wiretap channel model $\sW$ was introduced by Wyner in 1975 \cite{wyner}. 
In this model, there are two receivers $Y,Z$ and a single transmitter $X$. The transmitter aims at 
sending messages to Receiver 1 through a communication channel $W_1$. The information sent from $X$ to
$Y$ is also received by Receiver 2 through another channel $W_2$. The transmission problem in the system
$\sW(W_1,W_2)$ calls for designing a coding system that supports communication between $X$ and $Y$
in a way that is both reliable and secure. The reliability requirement is the usual one for 
communication systems, namely, that the error probability of decoding the information by $Y$ be made
arbitrarily low by increasing the block length of the encoding. At the same time, the transmission 
needs to be made secure in the sense that the information extracted by Receiver 2 about the message of
$X$ approaches zero as a function of the block length. 

To describe the problem in formal terms, denote the input alphabet of the transmitter by 
 ${\mathscr X}$, and the output alphabets of the channels $W_1$ and $W_2$ by
${\mathscr Y}$ and ${\mathscr Z}$, respectively.
%$\mathscr{X}^N$ refers to the $N$-dimensional cartesian product of $\mathscr{X}$.
%$\mathscr{Y}^N$ and $\mathscr{Z}^N$ have a similar meaning.
%We will first make some auxiliary definitions before describing
%the wiretap channel problem.
The messages that the transmitter can convey to Receiver 1 form a finite set denoted below by
$\mathscr{M}.$ For transmission over the channel the message is encoded using a mapping
$f:\mathscr{M}\to \mathscr{X}^N,$ where $\mathscr{X}^N$ is an $N$-fold repetition of the input alphabet. 
We say that $f$ is a length-$N$ block encoder of the transmitter. Capacity-attaining schemes for the wiretap
channel rely on randomized encoding, i.e., a mapping that
sends $\cM$ to a probability distribution on $\cX^N.$ In other words, the message $m\in \cM$
is encoded  as a sequence $x^N\in \cX^N$ with probability $f(x^N|m),$ and the encoder is defined
as a matrix of conditional probabilities $(f(x^N|m))_{m\in\cM}^{x^N\in\cX^N}.$

The decoder of Receiver 1 is a mapping $\phi: {\mathscr Y}^N \to {\mathscr M}$. We also denote by $P_{Y|X}$
and $P_{Z|X}$ the conditional distributions induced by the channels $W_1$ and $W_2$, respectively, and
define the induced distributions $P_{Y^N|X^N},P_{Z^N|X^N},$ where, for instance, 
$P_{Y^N|X^N}(y^N|x^N)=\prod_{i=1}^N P_{Y|X} (y_i|x_i)$, 
where $y_i$ and $x_i$ refer to the $i$-th symbol of the vectors $y^N$ and $x^N$, respectively.

\begin{definition}\label{def:1.1}
We say that the encoder-decoder pair $(f,\phi)$ \emph{gives rise to $(N,\epsilon)$-transmission} over
the wiretap channel $\sW$ if for all $m\in {\mathscr M}$
\begin{equation}
\sum_{x^N\in {\mathscr X}^N} f(x^N|m) P_{Y^N|X^N} (\phi(y^N)=m|x^N) \geq 1-\epsilon\label{reliability} 
\end{equation}
and
\begin{equation}
I(M;Z^N) \leq \epsilon, \label{security}
\end{equation}
where $M$ is the message random variable (RV) and $Z^N$ is the RV that corresponds to the observations of Receiver 2. 
\label{wiretap_definition}
\end{definition}

In Definition \ref{wiretap_definition}, Eq.~\eqref{reliability} represents the reliability of communication condition while \eqref{security} answers the security of transmission requirement.
We note that in many works on transmission with a secrecy constraint the security condition was formulated
in a more relaxed way, namely as the inequality
         \begin{equation}\label{eq:ws}
           (1/N)I(M;Z^N)<\epsilon.
         \end{equation}
This is particularly true about pre-1990s works in information theory, but also applies to some very recent
works on the wiretap channel, e.g., \cite{thangaraj,shamai,Koyluoglu2010}.  
However, as shown by Maurer in \cite{maurer1,maurer2}, this constraint does not fulfill the intuitive
security requirements in the system. More specifically,  it is possible to construct examples in which
inequality \eqref{eq:ws} is satisfied and at the same time Receiver 2 is capable of learning
$N^{1-\epsilon}$ out of $N$ bits of the encoding $x^N.$ In view of this,
Maurer suggested \eqref{security} as a better alternative to condition \eqref{eq:ws}. 
As a result, currently \eqref{eq:ws} is called the ``weak security constraint'' as opposed to the 
stronger constraint \eqref{security}. 
In this paper we design coding schemes that provide strong secrecy, so below we work only with condition \eqref{security}.

\vspace*{.1in}
The secrecy capacity of the wiretap channel is defined as follows.
\begin{definition} \label{def:1.2} The value $R>0$ is called an achievable rate for the wiretap channel $\sW$ if there exists a sequence
of message sets ${\mathscr M}_{N}$ and encoder-decoder pairs
$(f_N, \phi_N)$ giving rise to $(N,\epsilon_N)$ transmission with $\epsilon_N\to 0$ and
$\frac{1}{N}\log |{\mathscr M}_{N}|\to R$ as $N\to\infty$.
The secrecy capacity $C_s$ is the supremum of achievable rates for the wiretap channel.
\end{definition}

The following theorem provides an expression for $C_s$.
\begin{theorem}\label{thm:dg}
{\rm ({\cite{csiszar2}; see also \cite[p.411]{csiszar}})} The secrecy capacity of the wiretap channel $\sW$ equals
\begin{equation}
C_s=\max [I(V;Y)-I(V;Z)], \label{eq:cs}
\end{equation}
where the maximum is computed over all RVs $V,X,Y,Z$ such that the Markov condition
$V\to X\to Y,Z$ holds true, and such that $P_{Y|X}=W_1$, $P_{Z|X}=W_2$. 
\end{theorem}

%In this paper, we show that $C_s$ is achievable with strong secrecy via polar codes.
%We do not make any assumptions on the wiretap channel, so our result holds for any wiretap channel.

While most general constructive coding schemes for the wiretap channels rely on polar codes, 
there were some constructive solutions even before the publication of  Ar{\i}kan's seminal work  \cite{arikan2009}.
At the same time, these schemes applied only to some special cases of the channels $W_1,W_2.$ For instance
the case when $W_1$ is noiseless and $W_2$ is a binary erasure channel was addressed in
\cite{suresh,thangaraj} which show that in this case the capacity $C_s$ can be achieved using low-density
parity-check codes. The results in \cite{thangaraj} are
based on the weak security assumption while strong security is considered in \cite{suresh}.
Moreover, \cite{thangaraj} extends the construction to the cases when 
both $W_1$ and $W_2$ are erasure channels, and when $W_1$ is noiseless and $W_2$ is a
binary symmetric channel. 

Another special case of the wiretap channel relates to the combinatorial version of the erasure channel
(the so-called wiretap channel of type II) in which 
Receiver 2 can choose to observe any $t$ symbols out of
$N$ transmitted symbols. Constructive capacity-achieving solutions for this case are based on 
MDS codes \cite{wei} or extractors \cite{cheraghchi}.

In \cite{bellare}, it is shown that $C_s$ is achievable with strong security using invertible extractors,
if both $W_1$ and $W_2$ are binary symmetric channels. Both encoding and
decoding algorithms in \cite{bellare} have polynomial complexity. Moreover, \cite{bellare} also claims that
its proof method can be easily extended to other wiretap channels
as long as both $W_1$ and $W_2$ are symmetric. 

After the introduction of polar codes by Ar{\i}kan, achieving $C_s$ via polar coding has been 
considered by different works, mostly under the degradedness assumption. Recall that 
a channel $W_2:\cX\to\cZ$ is called {\em degraded} with respect to a channel $W_1:\cX\to\cY$ if there 
exists a stochastic $|\cY|\times|\cZ|$ matrix $P_{Z|Y}(z|y)$ such that for all $x\in \cX,z\in \cZ$
    \begin{equation}\label{eq:degraded}
P_{Z|X}(z|x)=\sum_{y\in{\mathscr Y}} P_{Y|X}(y|x) P_{Z|Y}(z|y).
    \end{equation}  
The wiretap channel $\cW$ is called degraded if the channel to the eavesdropper is degraded with respect to the
main channel. In this case Theorem \ref{thm:dg} affords a simpler formulation because there is no need in the auxiliary RV $V$.
Namely, in the degraded case the secrecy capacity equals \cite[Probl.~17.8]{csiszar}
   \begin{equation}\label{eq:dgr}
     C_s=C(W_1)-C(W_2)
   \end{equation}
(this specialization is true under more general assumptions, but we will not need them below).

Communication over degraded wiretap channels using polar codes was considered in a number of papers, notably, 
\cite{Mahdavifar2011,shamai,and10,Koyluoglu2010}. The main result of these works is that 
secrecy capacity \eqref{eq:dgr} can be attained under 
the weak security constraint. We note that the degraded case is easier to handle with polar codes because
of the specific nature of the polar codes construction (more on this below in Sec.~\ref{sect:closer}).
Another step was made by \cite{sasoglu} which suggested a polar coding scheme that attains the rate $C_s$ 
of a symmetric degraded wiretap channel $\sW$ under the strong security requirement \eqref{security}.
More details about the results of \cite{sasoglu} are given in Sec.~\ref{sect:closer} below. 

\vspace{.05in} The problem of attaining secrecy capacity of the general wiretap channel 
\eqref{eq:cs} under the strong secrecy condition and without the degradedness assumption
 was further studied in \cite{sutter}.
 \remove{
A polar coding scheme suggested in this work attains a transmission rate of
 $\max_{p_{X}(x)}[H(X|Z)-H(X|Y)].$ 
Clearly $H(X|Z)-H(X|Y)\le C_s$ for all $X$ since one can take
$V=X$ in the Markov chain $V\to X \to Y,Z,$ which appears in Theorem \ref{thm:dg}.
It is not immediately clear for which channels the result of \cite{sutter} actually attains the
secrecy capacity of $\sW.$}The coding scheme employed in \cite{sutter} relies on two nested layers of the polarizing
transform. The decoder for the second (outer) layer works with the probability distribution 
generated by the first decoder, which is not easily computable. Thus, the low complexity decoding claim
of the construction made in \cite{sutter} is not supported by the known decoding procedures for polar codes.
For these reasons the construction in \cite{sutter} does not resolve the question of constructing
a capacity-achieving polar coding scheme for the nondegraded case of wiretap channels.

In related works \cite{wilde,renes} the problem of constructing capacity achieving schemes for wiretap channels
was addressed for the case of quantum channels. The constructions suggested in these works attain  symmetric
secrecy capacity of quantum wiretap channels. %and this clearly applies to the classical case as well.
%At the same time 
These constructions require a shared secret key between the transmitter and Receiver 1. This
requirement seems to be intrinsic to polar code constructions for this problem including our work. However the
constructions in \cite{wilde,renes} require a positive-rate shared key, 
{whereas the bit partition introduced by our encoding setup makes it possible to have a
vanishing rate shared key, as explained in Sec. \ref{sect:wt}.}

To summarize, to the best of our knowledge the question of constructing capacity-achieving polar coding schemes for
the non-degraded wiretap channel with strong secrecy is an open problem.\footnote{
A concurrent study \cite{chou14}, posted
after the completion of this work, also contains a solution of the problems considered in this paper, including the
general wiretap channel. The transmission scheme and the proof methods in \cite{chou14} are different from our work. Another recent paper, \cite{wei14}, also devoted to the general wiretap channel, focuses on the weak security
requirement. A more detailed discussion of the relation of our work and \cite{chou14} appears in Sec. \ref{conclusion}.} 

It is this problem that we aim to solve in this paper by removing the degradedness assumption
\eqref{eq:degraded}. We also do not assume that either of the channels $W_1,W_2$ is symmetric.
The main idea of our work is to exploit the Markov chain conditions intrinsic to
secure communication problems using polar codes. In Sec. \ref{sect:wt} we propose a polar coding scheme that attains the secrecy capacity \eqref{eq:cs} under the strong security assumption. Both the encoding and decoding complexity estimates of our construction are $O(N\log N)$, where $N$ is the length of the encoding. 
In Sec. \ref{csiszar-korner} we generalize our construction to cover the case when a part of transmitter's message is public, i.e., is designed to be conveyed both to Receivers 1 and 2. This model, called a broadcast channel with confidential messages, is in fact the principal model in the founding work of Csisz{\'a}r and K{\"o}rner \cite{csiszar2} on this topic.

Apart from the basic polar coding results \cite{arikan2009}, our solution of the described problems relies on the previous work on the wiretap channel \cite{sasoglu}, the polar coding scheme for the broadcast channel
of \cite{Mondelli14}, and the construction of polar codes for general memoryless channels in \cite{Honda13}. 
A new idea introduced in our solution is related to a stochastic encoding scheme that emulates the
random coding proof of the capacity theorem in \cite{csiszar2}, whereby  
polarization is used for the values of the auxiliary random variable $V$ in Theorem \ref{thm:dg}, followed by a stochastic encoding into a channel codeword. Another insight, which is particularly useful for the
broadcast channel result in Sec.~\ref{sect:bcc}, is related to a partition of the 
coordinates of the transmitted block that enables simultaneous decoding of different parts of the transmitted
message by both receivers, whereby the decoder of polar codes is used by the receivers according to their high- and low-entropy bits. It becomes possible to show that the receivers recover the bits designed to communicate with 
each of them with high probability, and that the secret part of the message is not accessible to the 
unintended recipient.

%\vspace*{-.2in}
\section{Preliminaries on Polar Coding}\label{sec:prelim}

We begin with recalling basic notation for polar codes and then continue with 
the scheme for capacity-achieving communication on discrete binary-input channels.

Let $W$ be a binary-input channel with the output alphabet ${\mathscr Y},$ input alphabet
${\mathscr X}=\{0,1\},$ and the conditional probability distribution $W_{Y|X}(\cdot|\cdot)$, 
having capacity $C(W):=\max_{P_X}I(X;Y)$. The {\em symmetric capacity} $I(W)$ is the value of mutual information
$I(X;Y)$ when $P_X(0)=P_X(1)=1/2.$
%of $W$ is defined as
%$$I(W)=\sum_{y\in {\mathscr Y}} \sum_{x\in {\mathscr X}}  \frac{1}{2} W_{Y|X}(y|x)  
%\log \frac{W_{Y|X}(y|x)}{\frac{1}{2}W_{Y|X}(y|0)+\frac{1}{2}W_{Y|X}(y|1)}.$$ } 

\remove{
Throughout the paper we denote the capacity and  the symmetric capacity of $W$ by 
$C(W)$ and $I(W)$, respectively.} We say that the channel
$W$ is symmetric if 
$W_{Y|X} (y|1), y\in {\mathscr Y}$ can be obtained from $W_{Y|X} (y|0), y\in {\mathscr Y}$  
through a permutation $\pi:\cY\to\cY$ such that $\pi^2=\text{id}.$ 
If $W$ is symmetric, then $I(W)=C(W).$

Given a binary RV $X$ and a discrete RV $Y$ supported on ${\mathscr Y}$,
define the Bhattacharyya parameter $Z(X|Y)$ as 
\begin{equation*}
Z(X|Y)=2\sum_{y\in{\mathscr Y}} P_Y(y) \sqrt{{P_{X|Y}(0|y) P_{X|Y}(1|y)}}.
\end{equation*}
The value $Z(X|Y), 0\le Z(X|Y)\le 1$ measures the amount of randomness in $X$ given $Y$ in the sense
that if it is close to zero, then $X$ is almost constant given $Y$, while if it is close to one, then $X$ is almost uniform on $\{0,1\}$ given $Y$.
The Bhattacharyya parameter $Z(W)$ of a binary-input channel $W$ is defined as
\begin{equation*}
Z(W)=\sum_{y\in {\mathscr Y}} \sqrt{W_{Y|X}(y|0) W_{Y|X}(y|1)}.
\end{equation*}
Clearly if $P_X(0)=P_X(1)=1/2,$ then $Z(X|Y)$ coincides with the value $Z(W)$ for the communication channel
$W: X\to Y.$ 

For $N=2^n$ and $n\in{\mathbb N}$, define the polarizing matrix (or the Ar{\i}kan transform matrix)
as $G_N=B_N F^{\otimes n}$, where $F=\text{\small{$\Big(\hspace*{-.05in}\begin{array}{c@{\hspace*{0.05in}}c}
    1&0\\[-.05in]1&1\end{array}\hspace*{-.05in}\Big)$}}$,
$\otimes$ is the Kronecker product of matrices, and $B_N$ is a 	``bit reversal'' permutation 
matrix \cite{arikan2009}. In \cite{arikan2009}, Ar{\i}kan showed that given a 
symmetric and binary input channel $W$, an appropriate subset of the rows of $G_N$ can be used as a generator matrix of a linear code that achieves the capacity of $W$ as $N\to\infty$.

\subsection{Symmetric Channel Coding}
\label{symmetric}
The material in this section and Sec. \ref{sect:asymmetric} is well understood, but it merits some space in the present paper because it helps us to define the terminology that is useful for the main results below.
Given a symmetric binary-input channel $W$, define the channel $W^N$ with input alphabet
$\{0,1\}^N$ and output alphabet ${\mathscr Y}^N$ by the conditional distribution
\begin{equation*}
W^N(y^N|x^N)= \prod_{i=1}^N W(y_i|x_i),
\end{equation*}
where $W(.|.)$ is the conditional distribution that defines $W$.
Define a combined channel $\widetilde{W}$ by the conditional distribution
\begin{equation*}
\widetilde{W}(y^N|u^N)= W^N(y^N|u^N G_N).
\end{equation*}
In terms of $\widetilde{W}$, the channel seen by the $i$-th bit $U_i, i=1,\dots,N$
(also known as the bit-channel of the $i$-th bit) can be written as
\begin{equation*}
W_i (y^N, u^{i-1}|u_i)= \frac{1}{2^{n-1}} \sum_{\widetilde{u}\in \{0,1\}^{n-i}} \widetilde{W} (y^N| (u_1^{i-1},u_i, \widetilde{u})),
\end{equation*}
where $u^{i-1}=(u_1,u_2,\dots,u_{i-1})$. %From Figure \ref{figure1a}, 
We see that $W_i$ is the conditional distribution of $(Y^N, U^{i-1})$ given $U_i$
provided that the channel inputs $X_i$ are uniformly distributed for all $i=1,\dots,N$.
\remove{
\begin{figure}[h]
\centering
\includegraphics[scale=0.9]{arikan.pdf}
\caption{Block diagram of the symmetric channel coding}\label{figure1a}
\end{figure}
}
The bit-channels thus defined are partitioned into good channels ${\mathscr G}_N (W,\beta)$
and bad channels ${\mathscr B}_N (W,\beta)$ based on the value of their Bhattacharyya parameters.
Bearing in mind our notation choices later in the paper, we denote them by
\begin{equation}\label{good-bad}
\begin{aligned}
 \cL_{X|Y}=\cL_{X|Y}(N)&= \{ i\in [N]: Z(W_i) \leq \delta_N \}\\
  \cH_{X|Y}=\cH_{X|Y}(N)&= \{ i\in [N]: Z(W_i)  > 1-\delta_N \} ,
\end{aligned}
\end{equation}
respectively, where $[N]=\{1,2,\dots,N\}$ and $\delta_N \triangleq 2^{-N^{\beta}}, \beta\in (0,1/2).$
As shown in \cite{ari09a,korada}, for any symmetric binary-input channel $W$ and
any constant $\beta<1/2,$
      \begin{equation}\label{telatar}
   \begin{aligned}
     \lim_{N\to\infty} \frac{|\cL_{X|Y}(N)|}{N}&= C(W)\\
     \lim_{N\to\infty} \frac{|\cH_{X|Y}(N)|}{N}&= 1-C(W).
     \end{aligned}
   \end{equation}
Based on this equality, information can be transmitted over the good-bit channels while the remaining
bits are fixed to some values known in advance to the receiver (in polar coding literature they are called \emph{frozen bits}).  
The transmission scheme can be described as follows: % (Figure \ref{figure1a}):
A message of $k=|\cL_{X|Y}|$ bits is written in the bits $u_i, i\in \cL_{X|Y}.$ 
The remaining $N-k$ bits of $u^N$ are set to 0. This determines the sequence $u^N$ which
is transformed into $x^N=u^N G_N,$ and the vector $x^N$ is sent over the channel. Denote by $y^N$
the sequence received on the output.
The decoder finds an estimate of $u^N$ by computing the values $\hat u^i, i=1,\dots,N$ as follows:
\begin{align}\label{eq:ML}
\hat{u}_i=
\begin{cases}
\argmax_{u\in \{0,1\}}W_i (y^N,\hat{u}^{i-1}|u), &\text{if} \, i\in \cL_{X|Y},  \\
0, &\text{if} \, i\in \cH_{X|Y}.
\end{cases}
\end{align}
The results of \cite{arikan2009,ari09a} imply the following upper bound on the error probability $P_e=\Pr(\hat u^N\ne u^N):$
\begin{equation}\label{eq:ep}
P_e \leq \sum_{i\in \cL_{X|Y}} Z(W_i)  \leq N 2^{-N^{\beta}} \leq 2^{-N^{\beta'}},
\end{equation}
where $\beta$ is any number in the interval $(0,0.5)$ and $\beta'=\beta'(N)<\beta.$ 

This describes the basic construction of polar codes \cite{arikan2009} which  
attains symmetric capacity $I(W)$ of the channel $W$ with asymptotically vanishing error rate.

\begin{remark}\label{rem:oN}
There is a subtle point about the limit relations in \eqref{telatar}. Even though asymptotically the bit channels are
either good or bad, it is not true that $\cL_{X|Y}^c=\cH_{X|Y}$ because there is a nonempty
subset of indices $\cL_{X|Y}^c\backslash\cH_{X|Y}$ of cardinality $o(N)$ that is
neither good nor bad. This distinction has no import for the simple situation of transmitting over $W$, but leads to
complications in the multi-user systems considered below; see, e.g., \eqref{eq:p} in Sec.~\ref{sect:wt}. \remove{We use the terms ``not good,'' ``not bad'' to describe the corresponding subsets of bit channels.}\end{remark}

\subsection{General Channel Coding}
\label{sect:asymmetric}
Let $W$ be a binary-input discrete memoryless channel $W:\cX\to\cY$ and let
$P_X$ be the capacity achieving distribution of $W$. 
%In the case of uniform $P_X$, the
%symmetric capacity and capacity of $W$ coincide, i.e.,  $C(W)=I(W)$ holds. For that case, Ar{\i}kan showed that 
%the capacity of $W$ can be achieved by the construction given in Section \ref{symmetric}. This capacity achieving
If $P_X$ is not uniform, then the basic scheme attains a transmission rate $I(W)$ which is less than $C(W).$
This scheme was extended by Honda and Yamamoto \cite{Honda13} to cover the case of arbitrary distributions $P_X$, 
attaining the rate $C(W)$ for general binary-input channels.

To explain the idea in \cite{Honda13}, for a given block length $N$ define the sets
\begin{equation}\begin{aligned}
{\mathscr H}_X&=\{i\in [N]: Z(U_i|U^{i-1})\geq 1-\delta_N\} \\
{\mathscr L}_X&=\{i\in [N]: Z(U_i|U^{i-1})\leq \delta_N\} \\
{\mathscr H}_{X|Y}&=\{i\in [N]: Z(U_i|U^{i-1},Y^{N})\geq 1-\delta_N\} \\
{\mathscr L}_{X|Y}&=\{i\in [N]: Z(U_i|U^{i-1},Y^{N})\leq \delta_N\}
\end{aligned} \label{eq:L}
\end{equation}
where $U^N,X^N,Y^N$ have the same meaning as above. %shown in Figure \ref{figure1a}. 
It can be shown \cite[Theorem 1]{Honda13} that 
\begin{align*}
\lim_{N\to\infty} \frac{1}{N} |{\mathscr H}_X|&= H(X) \\
\lim_{N\to\infty} \frac{1}{N} |{\mathscr H}_{X|Y}|&= H(X|Y).
\end{align*}
Two special cases of these relations were proved earlier, see
\cite[Theorem 19]{korada} for uniform $U^N$ and \cite{ari09a} for a fixed (small) value of $\delta$.
% For the uniform input $U^N$ the result is proved in \cite[Theorem 19]{korada},
%and the special case when $\delta$ is a fixed small number is proved in \cite{ari09a}.}

We note that the set ${\mathscr L}_{X|Y}$ is the set of good bit channels defined 
in \eqref{good-bad}. Unlike the case of uniform $P_X$, it is not possible to use all of these channels to transmit 
information over $W.$ 
This is because if $i\notin {\mathscr H}_X$, then $U_i$ cannot be used to carry information conditioned
on previous bits $U^{i-1}$. Hence \cite{Honda13} argued that 
the set of information indices should be chosen as ${\mathscr I}\triangleq{\mathscr H}_X \cap {\mathscr L}_{X|Y}$
rather than ${\mathscr L}_{X|Y}$.

Since ${\mathscr H}_{X|Y}\subseteq {\mathscr H}_X$ and the number
of indices that are neither in ${\mathscr H}_{X|Y}$ nor
in ${\mathscr L}_{X|Y}$ is $o(N)$, we have
  $$ 
  \lim_{N\to\infty} \frac{1}{N} |{\mathscr I}|= \lim_{N\to\infty} \frac{1}{N} ({|{\mathscr H}_X|- |{\mathscr H}_{X|Y}|})= C(W),
  $$
i.e., transmitting the information using the bits $U_i,i\in \cI$ attains the capacity of the channel $W.$

The code construction in \cite{Honda13} makes use of the following partition of the coordinate set $[N]$:
   \begin{equation}\label{part}
\left.
\begin{array}{l}
{\mathscr F}_r={\mathscr H}_X \cap {\mathscr L}^c_{X|Y}\\
{\mathscr F}_d={\mathscr H}_X^c\\[.05in]
{\mathscr I}={\mathscr H}_X \cap {\mathscr L}_{X|Y}
\end{array}\right\}
    \end{equation}
where the superscript $^c$ refers to the complement of the subset in $[N].$ In terms
of this partition, the encoding is done as follows. The information bits are stored in
$\{u_i, i\in {\mathscr I}\}.$ As for the bits in the subset $\{i\in \cF_r \cup {\mathscr F}_d \},$ \cite{Honda13}
suggested to sample their values from the distribtion $P_{U_i|U^{i-1}}.$ These values are shared with the 
receiver similarly to the ``frozen bits'' of the symmetric scheme of \cite{arikan2009}.
Once $u^N$ is determined, the transmitter finds $x^N=u^N G_N$ and sends it over
the channel. 

The receiver uses the following successive decoding function: for $i=1,2,\dots,N$ let
  \begin{equation}\label{eq:scd1}
  \begin{aligned}
\hat{u}_i=\begin{cases}
\argmax_{u\in\{0,1\}} {P}_{U_i|U^{i-1},Y^{N}}(u|\hat{u}^{i-1},y^{N}), 
 i\in {\mathscr I}\\
 u_i, \hspace{1.8in} i\in  {\mathscr F}_r \cup {\mathscr F}_d .\\
\end{cases}
   \end{aligned}
   \end{equation}
Note that this rule represents a MAP decoder for the $i$th subchannel, which for the symmetric case coincides with the ML decoder
rule \eqref{eq:ML}.

The probability of decoding error can be bounded above similarly to \eqref{eq:ep}: 
    \begin{equation}\label{eq:perr}      
P_e\leq \sum_{i\in {\mathscr I}} Z(U_i|U^{i-1},Y^{N})\leq N 2^{-N^{\beta}} \leq 2^{-N^{\beta'}}
    \end{equation}
where the parameters have the same meaning as before.

Moreover, \cite{Honda13} argues that there exists a set of deterministic maps
  $\lambda_i: \{0,1\}^{i-1}\to \{0,1\}, i\in  {\mathscr F}_r \cup {\mathscr F}_d$ such that \eqref{eq:perr} holds true, stating
  the decoding rule in the form
 \begin{equation}\label{eq:scd3}
  \begin{aligned}
\hat{u}_i=\begin{cases}
\argmax_{u\in\{0,1\}} {P}_{U_i|U^{i-1},Y^{N}}(u|\hat{u}^{i-1},y^{N}), 
 i\in {\mathscr I}\\
 \lambda_i (\hat{u}^{i-1}), \hspace{1.4in} i\in  {\mathscr F}_r \cup {\mathscr F}_d.\\
\end{cases}
   \end{aligned}
   \end{equation}
The mappings $\lambda_i$ are shared between the transmitter and the receiver prior to communication.
We shall use this form of the transmission scheme and the decoder throughout our paper.

This completes the description of the capacity-achieving transmission scheme of \cite{Honda13}. We will rely in part
on these ideas in our construction of a coding scheme for the wiretap channel.

\section{A Closer Look at Prior Works on Polar Coding for the Wiretap Channel}\label{sect:closer}
To explain our proposal we will first discuss some of the schemes available in the literature.
We begin with the transmission scheme of \cite{Mahdavifar2011} (see also \cite{and10,Koyluoglu2010,shamai}). 
As already remarked, these works are concerned with the special 
case when the channel $W_2$ is degraded with respect to $W_1$ and aim to attain the rate value \eqref{eq:dgr} with weak secrecy.
Let $X^N$ be a random uniform vector over $\{0,1\}^N$.
Similarly to ${\mathscr H}_{X|Y}$ and ${\mathscr L}_{X|Y}$ given by \eqref{eq:L}, define the following subsets of indices:
    \begin{align*}
{\mathscr H}_{X|Z}&=\{i\in [N]: Z(U_i|U^{i-1},Z^{N})\geq 1-\delta_N\}\\
{\mathscr L}_{X|Z}&=\{i\in [N]: Z(U_i|U^{i-1},Z^{N})\leq \delta_N\} 
   \end{align*}
where $U^N=X^N G_N$, and $Z^N$ is the output that Receiver 2 observes when the transmitter
sends $X^N$. 
Partition the set $[N]$ as follows:
\begin{equation}
\begin{aligned}
{\mathscr R}&= {\mathscr L}_{X|Z} \\
{\mathscr I}&= {\mathscr L}_{X|Y}\backslash {\mathscr L}_{X|Z} \\
{\mathscr B}&= {\mathscr L}_{X|Y}^c.
\end{aligned}\label{eq:mp}
\end{equation}
Note that the degradedness assumption \eqref{eq:degraded} implies the inclusion ${\mathscr L}_{X|Z}\subseteq {\mathscr L}_{X|Y}$.
The coding scheme for the wiretap channel relies on the partition \eqref{eq:mp} and is summarized in Figure \ref{figure1}. \remove{
If $W_2$ is degraded with respect to $W_1$, then it follows that 
${\mathscr L}_{X|Z}\subseteq {\mathscr L}_{X|Y}$, ${{\mathscr R}}_1={\mathscr L}_{X|Z}$,
and ${{\mathscr R}}_2=\emptyset$, ${{\mathscr B}}={\mathscr L}_{X|Y}^C$. 
So, for the degraded case the sets ${{\mathscr I}}, {{\mathscr R}}_1, {{\mathscr B}}$
partition $[n]$, as shown by Figure \ref{figure1}.} 
The information is stored in the bits $u_i, i\in {\mathscr I}$.
The bits in the coordinates in $\cR$ are chosen randomly while the bits in $\cB$ form a subset of the frozen 
bits. 

Attainability of the rate \eqref{eq:dgr} using this coding scheme is proved in the cited papers.
An essential remark here is that
the bits $u_i, i\in {\mathscr R}$ are randomly selected because fixing their values contradicts even the 
weak security constraint, let alone the stronger one.

\begin{figure}[t]
\centering
\includegraphics[scale=0.7]{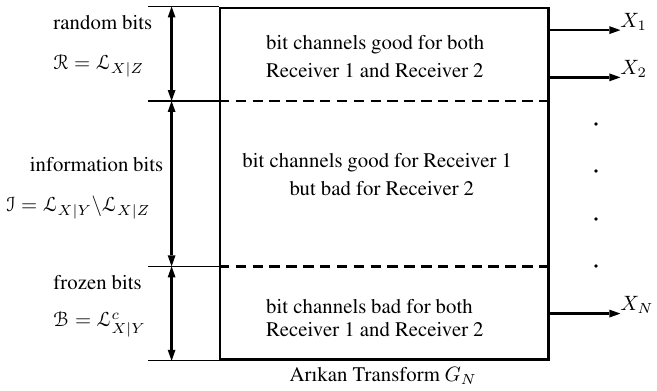}
\caption{Block diagram of the coding scheme in \cite{Mahdavifar2011}. Good bit channels and bad bit channels 
are as defined by \eqref{good-bad}.}\label{figure1}
\end{figure}

We note that generally ${\mathscr H}^c_{X|Z} \not\subset {\mathscr L}_{X|Y},$ and even though the number of 
coordinates in ${\mathscr L}_{X|Y}^c\cap{\mathscr H}^c_{X|Z}$ behaves as $o(N),$ this constitutes
an obstacle to achieving strong security.
To bypass it, \cite{sasoglu} uses a different partition of the coordinates, namely
\begin{equation}
 \begin{aligned}
 {\widetilde{\mathscr I}}&= {\mathscr L}_{X|Y}\cap {\mathscr H}_{X|Z} \\
{\widetilde{\mathscr B}}&= {\mathscr L}^c_{X|Y} \cap {\mathscr H}_{X|Z}\\
{\widetilde{\mathscr R}}_1&= {\mathscr L}_{X|Y}\cap {\mathscr H}^c_{X|Z}\\
{\widetilde{\mathscr R}}_2&= {\mathscr L}_{X|Y}^c\cap{\mathscr H}^c_{X|Z}.
\end{aligned}\label{sasoglu_partition}
\end{equation}
\remove{\begin{equation}
\begin{aligned}
{\widetilde{\mathscr R}}_1&= {\mathscr L}_{X|Y}\cap {\mathscr H}^c_{X|Z}\\
{\widetilde{\mathscr R}}_2&= {\mathscr H}^c_{X|Z} \backslash {\mathscr L}_{X|Y}\\
{\widetilde{\mathscr I}}&= {\mathscr L}_{X|Y}\backslash {\mathscr H}^c_{X|Z} \\
{\widetilde{\mathscr B}}&= {\mathscr L}^c_{X|Y} \cap {\mathscr H}_{X|Z}.
\end{aligned}\label{sasoglu_partition}
\end{equation}}
Apart from transmitting the information, the coding scheme aims 
to convey the bits in $\widetilde\cR_2$ to Receiver 1 using the good indices of Receiver 1, at the same
time preserving the security requirement. This is accomplished using the ``chaining'' construction proposed 
in \cite{sasoglu}\footnote{The term ``chaining'' was introduced later in \cite{Mondelli14}.}
and shown in Fig.~\ref{figure2}. As the figure suggests, the bits in $\widetilde\cR_2(j)$ contained in block $j$
are transmitted over the channel as a part of the message of block $j-1,$ for all $j=2,\dots,m.$ This enables 
Receiver 1 to recover these bits reliably as a part of the successive decoding procedure for block $j,$ which is performed similarly to \eqref{eq:scd1}. At the same time,
because of the inclusion $\widetilde\cI\subset \cH_{X|Z},$ Receiver 2 does not
have the resources for their reliable decoding, which provides the desired security.
 
\begin{figure}[t]
\centering
\includegraphics[scale=0.65]{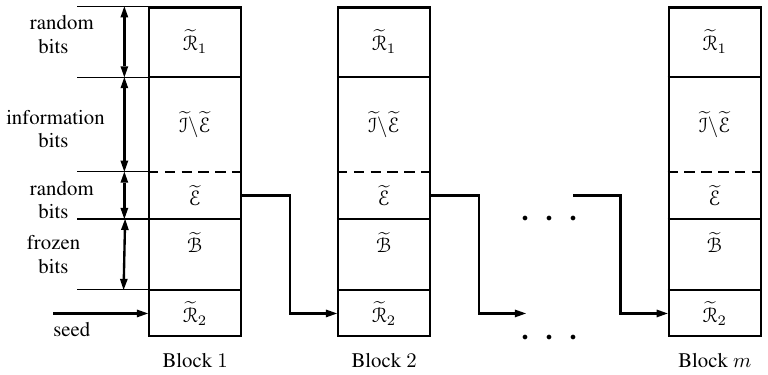}
\caption{Block diagram of the coding scheme in \cite{sasoglu}. 
}\label{figure2}
\end{figure}

The analysis of the transmission is performed based on $m$ blocks of $N$ bits as opposed to a single block. 
The seed for the transmission is provided by choosing $|{\widetilde{\mathscr R}}_2|=o(N)$ random bits which are
shared with Receiver 1 (more on this below).  In each of the blocks 1 to $m$, the bits 
indexed by the set ${\widetilde{\mathscr I}}\backslash\widetilde\cE$
are used to send the message. Here $\widetilde\cE$ is an arbitrary subset of the set $\widetilde{\mathscr I}$ of size 
${\widetilde{\mathscr R}}_2$ whose role is explained below.

The bits in ${\widetilde{\mathscr R}}_1$ are selected
randomly, and the bits $u_i, i\in{\widetilde{\mathscr B}}$ are frozen, i.e., assigned arbitrarily and shared
in advance with Receiver 1 (they may be also known to Receiver 2 without compromising secrecy).

The assignment of bits in the set ${\widetilde{\mathscr R}}_2$ in block $j$ depends on the block index.
In block 1 these bits are set equal to the message bits of the seed block. In block $j=2,\dots,m$ the bits
indexed by the set ${\widetilde{\mathscr R}}_2$ are set to be equal to the bits in the set $\widetilde\cE$
in block $j-1,$ representing the chaining procedure. 

Having formed the sequence $u^N(j)$ in block $j=1,\dots,m$, the encoder passes it through the polarizing transform and
transmits the sequence $x^N=u^N G_N$ over the channel. 
The only remaining problem is to convey to Receiver 1 the bits of $\widetilde \cR_2$ of the first block. 
This is done by performing the seed transmission of a block which encodes the $|\widetilde \cR_2|$ bits
using some error correcting code of length $N$. 
As claimed in \cite{sasoglu}, it is possible to choose such a code to fulfill the reliability and security
requirements because its rate can be made arbitrarily close to zero.
The fact that the seed code needs to encode only a small number $o(N)$ of message bits follows from the degradedness
assumption, which is therefore essential in this construction.

%The need to share a small number of random bits between the transmitter and the legitimate receiver arises in many 
%recent constructions based on polar codes \cite{wilde,renes}. As already mentioned, our scheme will also rely on sharing
%a random seed with the receiver. As we argue below, this assumption does not affect the transmission rate.

As shown in \cite{sasoglu}, this scheme satisfies both constraints 
\eqref{reliability} and \eqref{security} under the assumption that the channel to the eavesdropper is degraded with respect
to the channel $W_1.$ The rate of communication between the transmitter and Receiver 1 can be made arbitrarily close to the 
value $I(W_1)-I(W_2)$ since 
the assumption that $W_2$ is degraded with respect to $W_1$
ensures that $|{\widetilde{\mathscr R}}_2|=o(N)$, i.e., there is
no asymptotic loss in rate by removing the bits $\{u_i, i\in {\widetilde{\mathscr E}}\}$ from the message in order
to support the strong security condition.

\section{Polar Coding for the Wiretap Channel}
\label{sect:wt}

In this section, we show that secrecy capacity for the wiretap channel given by 
Theorem \ref{thm:dg} is achievable using polar codes. For this purpose, we consider
the RVs $V,X,Y,Z$ as described by Theorem \ref{thm:dg}, i.e., we assume some fixed
distributions $P_V, P_{X|V}$ and the conditional distributions $P_{Y|X}=W_1$, $P_{Z|X}=W_2$
that satisfy the Markov condition $V\to X \to Y,Z$ and maximize the expression in \eqref{eq:cs}.
Define the RV $T^{N}=V^{N}G_N,$ where $V^N$ denotes $N$ independent realizations of $V$.
The transformation $V^N\to T^N$ induces conditional distributions
$P_{T_i|T^{i-1}}$ derived from the corresponding distributions
of the RVs $V_i.$
Define the sets ${\mathscr H}_V$, ${\mathscr L}_V$, ${\mathscr H}_{V|Y}$, 
${\mathscr L}_{V|Y}$ as follows:
\begin{align*}
{\mathscr H}_V&=\{i\in [N]: Z(T_i|T^{i-1}) \geq 1-\delta_N\}  \\
\cL_V&= \{ i\in [N]: Z(T_i|T^{i-1}) \leq \delta_N\}\\
{\mathscr H}_{V|Y}&= \{i\in [N]: Z(T_i|T^{i-1},Y^{N}) \geq 1-\delta_N\}\\
{\mathscr L}_{V|Y}&= \{i\in [N]: Z(T_i|T^{i-1},Y^{N}) \leq \delta_N\}
\end{align*}
and define the sets ${\mathscr H}_{V|Z}$, ${\mathscr L}_{V|Z}$ analogously.
The cardinalities of these sets satisfy 
$\frac{1}{N} |{\mathscr H}_V|\to H(V), \frac{1}{N} |{\mathscr L}_{V|Y}|\to 1-H(V|Y), \frac{1}{N} |{\mathscr H}_{V|Z}|\to H(V|Z)
$ as $N\to\infty$ \cite[Theorem 1]{Honda13}.

Define a partition of $[N]$ into the following sets which will be used to describe the coding scheme\footnote{We use the notation $\cI,\cB$ in this section in the sense different from Sec.~\ref{sect:closer}. Since both uses are localized to their respective sections, this should not
cause confusion.}
\begin{equation}\label{eq:p}
\begin{aligned}
{\mathscr I}&= {\mathscr H}_{V}\cap {\mathscr L}_{V|Y}\cap {\mathscr H}_{V|Z} \\
{\mathscr B}&= {\mathscr H}_{V}\cap {\mathscr L}_{V|Y}^c \cap {\mathscr H}_{V|Z} \\
{\mathscr R}_1&={\mathscr H}_{V}\cap {\mathscr L}_{V|Y}\cap {\mathscr H}^c_{V|Z} \\
{\mathscr R}_2&= {\mathscr H}_{V}\cap {\mathscr L}^c_{V|Y}\cap {\mathscr H}^c_{V|Z} \\
{\mathscr D}&= {\mathscr H}^c_{V}.
\end{aligned}
\end{equation}
The partition of $[N]$ that thus arises is illustrated in Figure 3. 
It will be seen that the subsets ${\mathscr I}$, ${\mathscr R}_1$, ${\mathscr R}_2$, ${\mathscr B}$
in our coding scheme play the role similar to that of the analogously denoted subsets in \eqref{sasoglu_partition}.
Importantly, the cardinality of ${\mathscr R}_2$ is not $o(N)$ any more, which requires adjustments in the transmission scheme. 
Moreover, there is an extra randomness needed to determine
the sequence to be transmitted, as will be seen in the encoding algorithm below.

\begin{figure}[t]
\includegraphics[scale=0.6]{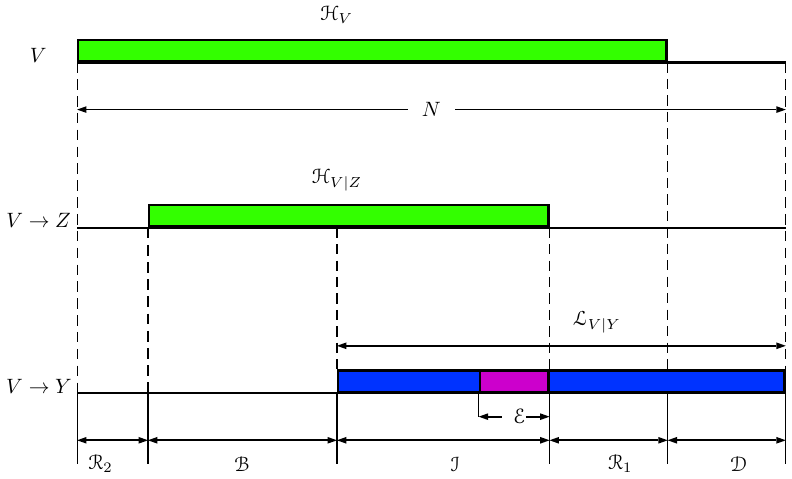}
\vspace*{.1in}
\begin{caption} {Partition of $N$ coordinates of the block for transmission over the wiretap channel $\sW$; see \eqref{eq:p}.
The highlighted part of the top block represents high-entropy coordinates for
the distribution $P_V$. Similarly, in the middle block we highlight the high-entropy coordinates
of the distribution $P_{V|Z}$ and in the bottom block the low-entropy
coordinates for the distribution $P_{V|Y}$. }
\end{caption}\label{fig:partition}
\end{figure}

%\vspace*{.1in}
{\bfit Encoding:} We build on the chaining idea of \cite{sasoglu}, connecting multiple blocks
in a cluster whose performance in transmission will attain the desired goals.
The cluster consists of a seed block and a number, $m$, of other blocks.
The seed block consists of $|{\mathscr R}_2|$ random bits. Even though the cardinality of the set ${\mathscr R}_2$
constitutes a nonvanishing proportion of $[N]$, the rate of the seed $|{\mathscr R}_2|/ m N$ can be
made arbitrarily small by choosing $m$ sufficiently large. (For example, one can
set $m=N^{\alpha}$ for some $\alpha>0$, and let $N\to\infty.$)

\begin{figure}[t]
\centering
\includegraphics[scale=0.55]{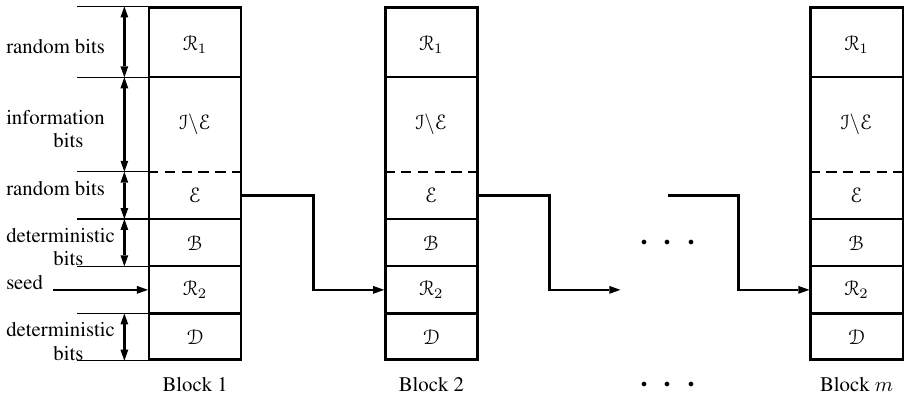}
\caption{Block diagram of the wiretap coding scheme: Forming the blocks $t^N(j),j=1,\dots,m.$}
\label{figure3}
\end{figure}
Let us describe the encoding and decoding procedures. The transmission is accomplished using $m$ blocks of
 length $N$ each and the seed block.
 Every block $t^N=t^N(j), j=1,\dots,m$ contains a group of almost deterministic bits, denoted by $\cD$ in Figure \ref{figure3}. The values of these bits are assigned according to a {family of deterministic rules} $\{\lambda_i, i\in\cD\}$ described in \eqref{eq:scd3}.
The bits $t_i,i\in\cB$ in each of the $m$ blocks are determined similarly, based on $\{\lambda_i, i\in\cB\}$.
\textcolor{black}{These rules are chosen the same for each block, and are shared with Receiver 1. 
Even if the rate $(|\cB|+|\cD|)/mN$ can be made arbitrarily small similarly to $|\cR_2|/mN$, it may be
the case that such a rule sharing requires a positive rate of secure communication. As a remedy to this
problem, the bits $t_i,i\in\cB$ can be chosen randomly and be the same for each block, and shared with
Receiver 1 secretly. This pre-shared randomness still allows one to maintain secrecy as shown by the results presented below.
We note that \cite{chou14} also relies on the same assumption for their proof of the secrecy condition.}

The remaining subsets of coordinates are filled as follows. For block 1, the bits in the set $\cR_2$ are assigned the value of the bits of the seed block, while for blocks $j=2,\dots,m$ these bits are set to be equal to the bits in $\cE(j-1)$ of block $j-1.$ 
Here, $\cE$ is a subset of $\cI$ having the same size as $\cR_2$.
The messages are stored in the bits indexed by $\cI\backslash\cE.$ The randomly
chosen bits in $\cE$ are written in the coordinates that are good for Receiver 1 and contained in the bad
(high-entropy) set of Receiver 2. These bits are transmitted to Receiver 1 in block $j$ and used for 
the decoding of the message contained in block $j+1,$ for all $j=1,2,\dots,m-1.$ Finally, the bits in 
$\cR_1$ are assigned randomly and uniformly for each of the $m$ blocks.
The diagram of the chaining construction for encoding is given in Figure \ref{figure3}.
Let $\cE(0)$ be the message sequence encoded in
the seed block, let $\cE(j), j=1,\dots,m$ be the corresponding sequences in the other blocks 
(see Fig.~\ref{figure3}), and let
$r(j), j=1,\dots,m$ the randomly chosen bits in each block. The encoding proceeds as follows:
In Block 1 we put
\begin{align}
t_i= \begin{cases}
\lambda_i(t^{i-1}(1)) &\text{if}\,\, i\in \cB\cup \cD \\
\cE_i(0) &\text{if}\,\, i\in \cR_2 \\
r_i(1) &\text{if}\,\, i\in \cI\cup \cR_1
\end{cases}
\end{align}
and in blocks $j=2,\dots,m$ we put
\begin{align}
t_i= \begin{cases}
\lambda_i(t^{i-1}(j)) & \text{if}\,\, i\in B\cup D \\
\cE_i(j-1) & \text{if}\,\, i\in R_2 \\
r_i(j) & \text{if}\,\, i\in I\cup R_1,
\end{cases}
\end{align}
where the family of mappings $\lambda_i$ is chosen from the ensemble
$P_{\Lambda_i}[\Lambda_i (t^{i-1})=1]=P_{T_i|T^{i-1}}(1|t^{i-1})$ for 
all $t^{i-1}\in \{0,1\}^{i-1}$.

Once the blocks $t^N(j), j=1,2,\dots,m$ are formed, we find $m$ sequences $v^N(j)=t^N(j)G_N$ by using the polarizing transform. 
Finally, given $v^{N},$ the codeword to be sent over the wiretap channel will be chosen as $x^{N}$
with probability $P_{X^{N}|V^{N}}(x^{N}|v^{N})=\prod_{i=1}^N P_{X|V}(x_i|v_i)$, where
$P_{X|V}$ is the conditional distribution induced by the joint distribution of the RVs $V$ and $X$.
This logic is suggested by the proof of the capacity theorem, Theorem \ref{thm:dg}, which
first considers ``transmitting'' the RV $V^N$ to the receivers, and then choosing $X^N$ so as to satisfy the Markov chain condition in the statement.

%This ensures that the output sequences of the channels $W_i,i=1,2$ follow the distributions $P_{Y|X}$ and $P_{Z|X}.$

\vspace*{.1in}{\bfit Decoding:}  

Let $y^N(1),\dots,y^N(m)$ be the sequences that Receiver 1 observes on the output of the channel $W_1.$
The decoding rule is as follows:
    \begin{equation}\label{eq:dec}
    \hat{t}_i= \begin{cases}
\lambda_i (\hat{t}^{i-1}),  \hspace{1.6in}\text{if}\,\, i\in  {\mathscr B}\cup {\mathscr D} \\
\argmax_{t\in \{0,1\}} P_{T_i|T^{i-1}Y^{N}}(t|\hat{t}^{i-1},y^{N}), \text{if}\,\, i\in {\mathscr I}\cup{\mathscr R}_1 \\
\cE_i(j-1), \hspace{1.6in}\text{if} \,\, i\in {\mathscr R}_2
\end{cases}
   \end{equation}
where $P_{T_i|T^{i-1}}$ and $P_{T_i|T^{i-1},Y^{N}}$ are the conditional distributions induced by the
joint distribution of the RVs $V^{N}$ and $Y^{N}$ (this rule is applied to each of the blocks $j=1,\dots,m$, and $j$ is mostly
omitted from the notation).
%%%%%%%%%%%%%%
\remove{
Denote by $E(0)$ the message sequence encoded in
the seed block, and by $E(j-1), j=2,\dots,l+1$ the sequence of bits $\{t_i, i\in {\mathscr E}\}$ contained in block $j-1;$  
see Fig.~\ref{figure3}.
Moreover, let $y^N(1),\dots,y^N(l)$ be the sequences that Receiver 1 has observed as $l$ blocks are transmitted.
 For $j=1,\dots,l$, the decoding rule is as follows:
\begin{equation}\label{eq:dec}
\hat{t}_i(j)= \begin{cases}
%t_i, &\text{if} \,\, i\in {\mathscr B}\\
\argmax_{t\in \{0,1\}} P_{T_i|T^{i-1}}(t|\hat{t}^{i-1}(j)), &\text{if}\,\, i\in {\mathscr D}\cup {\mathscr B} \\
\argmax_{t\in \{0,1\}} P_{T_i|T^{i-1},Y^{N}}(t|\hat{t}^{i-1}(j),y^{N}(j)), &\text{if}\,\, i\in {\mathscr I}\cup{\mathscr R}_1\\
E_i(j-1), &\text{if} \,\, i\in {\mathscr R}_2
\end{cases}
\end{equation}
where $P_{T_i|T^{i-1}}$ and $P_{T_i|T^{i-1},Y^{N}}$ are the conditional distributions induced by the
joint distribution of the RVs $V^{N}$ and $Y^{N}$. }
%%%%%%%%%%%%%%

Let us show that the described scheme attains the secrecy capacity of $\sW.$ Namely, the following is true.
\begin{theorem}\label{prop:wiretap} 
For any $\gamma>0$, $\epsilon>0$ and $N\to\infty$ it is possible to choose $m$ so that the transmission scheme described 
above attains the transmission rate $R$ that is within $\gamma$ of the secrecy capacity of $\sW$ \eqref{eq:cs}
and the information leaked to Receiver 2 satisfies the strong secrecy condition \eqref{security}.
\end{theorem}
\begin{proof}
Throughout the proof we assume that the RVs $V,X,Y,Z$ are as given in Theorem \ref{thm:dg} and denote by $P_{VXYZ}$ their joint 
distribution. The distribution $P_{V^NX^NY^NZ^N}=\prod_{i=1}^N P_{VXYZ}(v_i,x_i,y_i,z_i)$ refers to $N$ independent repetitions
of the RVs.

The rate of the proposed coding scheme is
\begin{equation*}
\frac{m (|{\mathscr I}|-|{\mathscr E}|)}{m N+|{\mathscr R}_2|}= \frac{ (|{\mathscr H}_{V}\cap {\mathscr L}_{V|Y}|-|{\mathscr H}_{V}\cap {\mathscr H}^c_{V|Z}|)}{ N+ |{\mathscr R}_2|/m},
\end{equation*}
which approaches $I(V;Y)-I(V;Z)$ as $N,m\to\infty,$
where the term $|R_2|$ in the denominator is due to the shared seed
for block 1.
According to Theorem \ref{thm:dg}, this is the target rate that we want to achieve for given $V$ and $X$ satisfying
$V\to X \to Y,Z$ and $P_{Y|X}=W_1$, $P_{Z|X}=W_2$.

Now let us prove the reliability and security conditions.
Let us introduce the following RVs: Let $M^{m}=(M_1,M_2,\dots,M_m)$ correspond to the sequence of message bits
 $\{t_i, i\in {\mathscr I} \backslash {\mathscr E}\}$ transmitted in blocks $1,\dots,m,$ and let
$Z^m= (Z^N(1),\dots,Z^N(m))$ be a sequence of observations of Receiver 2 as a result of the transmission of the $m$ blocks. 
{Further, let  $E_j$ correspond to the bits contained in the subset $\cE(j), j=1,\dots,m.$}

\vspace*{.1in}
{\em Reliability:} The claim of low error probability for Receiver 1 follows from the results of
\cite{Honda13}. Since our communication scheme is more complicated compared to \cite{Honda13}, we give 
some additional details. 

Since we know the distribution $P_{V^N},$
we can compute the distributions $P_{T_i|T^{i-1}}, i=1, \dots, N.$ Now assume that the assignments of the bits indexed by $\cB$ and 
$\cD$ are done randomly by sampling from the distribution $P_{T_i|T^{i-1}},$ for each of the blocks $1,\dots,m.$
Let $Q_{V^NX^NY^NZ^N}$ be the joint distribution of the correspoding sequences arising from this assignment.
%coming from the encoding and transmission
%scheme if the deterministic rules $\{\lambda_i, i\in \cB\}$ and $\{\lambda_i, i\in \cD\}$
%are replaced by random bit assignments based on the distribution $P_{T_i|T^{i-1}}$.
Denote by $\|\cdot\|$ the $l_1$ distance between the distributions. From the proof of Lemma 1 in \cite{Honda13} it follows that
\begin{equation}
\|P_{V^NX^NY^NZ^N}-Q_{V^NX^NY^NZ^N}\| \leq N 2^{-N^{\beta}}
\label{proximity}
\end{equation}
holds for all the $m$ blocks of transmission. Moreover, since the message bits 
are entirely contained in the set of good bits for channel $W_1$, 
the probability of error is bounded by
   \begin{gather*}
   \|P_{V^NX^NY^N}-Q_{V^NX^NY^N}\|+\sum_{i\in {\mathscr I}} Z(T_i| T^{i-1}, Y^N) \nonumber\\ \leq 2N 2^{-N^{\beta}}
   \end{gather*}
for each individual block (see \cite{Honda13}, Eqns (59)-(60)). Therefore,
there exists a family of deterministic rules $\{\lambda_i, i\in \cB\cup\cD\}$
such that the overall error probability for the successive decoding procedure \eqref{eq:dec} is at most
$2mN2^{-N^{\beta}}, \beta\in (0,1/2)$. 
We conclude that the probability that Receiver 1 decodes the information bits correctly approaches $1$ as $N$ tends to infinity.

\remove{
Namely, since the message bits are entirely contained in the set of good bits for channel $W_1$, 
the successive decoding procedure \eqref{eq:dec} will recover their values with error at most
$mN2^{-N^{\beta}}, \beta\in (0,1/2)$. 
We conclude that the probability that Receiver 1 decodes the information bits correctly approaches $1$ as $N$ tends to infinity. 
}

\vspace*{.1in}

{\em Security:} We will show that condition \eqref{security} is fulfilled for the sequence of $m$ blocks of transmission.
{For that purpose, we will first prove the following lemma. 
\begin{lemma} Let $\cA\subset \cI$ be a subset of coordinates, and let
let $T[{\mathscr A}]=\{t_i, i\in{\mathscr A}\}$ and $T[\cI\backslash {\mathscr A}]=\{t_i, i\in\cI\backslash{\mathscr A}\}.$ 
%where ${\mathscr A}$ is any subset of $\cI$.
Then
\begin{equation}
I(T[\cI\backslash{\mathscr A}]; T[{\mathscr A}], Z^N)=O(N^3 2^{-N^{\beta}}). \label{lemma_ineq_wiretap}
\end{equation}
\label{security_lemma_wiretap}
\end{lemma}
\begin{proof}
By definition in \eqref{eq:p} we have the inclusion $\cI\subseteq \cH_{V|Z}$. 
Let us label the indices in $\cI$ as $a_1,a_2,\dots,a_{|\cI|}$,
and assume that $a_1<a_2<\dots <a_{|\cI|}$. 
Using the inequality $Z(X|Y)^2\le H(X|Y)$ \cite{Arikan2010}, we obtain the estimate
\begin{align}
H_P(T_{a_i}|T^{a_i-1}, Z^N)&\triangleq H_P(T_{a_i}|T_1,\dots,T_{a_i-1},Z^N) \nonumber\\ &\ge 1-2\delta_N=1-2^{-N^{\beta}+1} 
\label{prox_review1}
\end{align}
for all $i=1,2,\dots,|\cI|$, where $H_P$ refers to the entropy under the distribution $P_{V^NX^NZ^N}$.
Our aim is to find an estimate on the mutual information in \eqref{lemma_ineq_wiretap}
under the distribution $Q$. To find an upper bound for the entropy $H(T_{a_i}|T^{a_i-1}\, Z^N)$ (computed under $Q$), we use a standard estimate (e.g., \cite[Theorem 17.3.3]{Cover06}),
and write
\begin{align*}
&|H(T_{a_i}|T^{a_i-1},Z^N)- H_P(T_{a_i}|T^{a_i-1}, Z^N)| \nonumber\\ &\leq - \|P_{V^NX^NZ^N}-Q_{V^NX^NZ^N}\| \nonumber\\ &\quad\times\log % \nonumber\\ &\times \log 
\frac{\|P_{V^NX^NZ^N}-Q_{V^NX^NZ^N}\|}{|{\mathcal Z}|^N 2^N }
\end{align*}
where $|{\mathcal Z}|$ refers to the alphabet size for a single observation of the eavesdropper. Then, we get from \eqref{proximity} and \eqref{prox_review1} that
\begin{align}
H(T_{a_i}|T^{a_i-1}, Z^N) &\geq 1-O(N^2 2^{-N^{\beta}})- O(N^{({\beta}+1)} 2^{-N^{\beta}}) \nonumber\\ &= 1-O(N^2 2^{-N^{\beta}}). \label{lemma_ineq1_wiretap}
\end{align}
Observe that \eqref{lemma_ineq1_wiretap} implies
\begin{align}
H(T_{a_i}|T_{a_1}, T_{a_2}, \dots, T_{a_{i-1}}, Z^N) &\ge H (T_{a_i}|T^{a_i-1}, Z^N)\nonumber\\  &\ge 1-O(N^2 2^{-N^{\beta}})  \label{lemma_ineq2_wiretap}
\end{align}
for all $i\in \cI.$ 
\newline
Then we obtain
   \begin{align}
H(T[\cI\backslash {\mathscr A}]\,| T[{\mathscr A}], Z^N)&= H( T[\cI]\,| Z^N)- H(T[{\mathscr A}]\,| Z^N) \nonumber\\
&\hspace{-0.5in}\geq H( T[\cI]| Z^N)- |{\mathscr A}| \nonumber\\
&\hspace{-0.5in} =\sum_{i=1}^{|\cI|} H(T_{a_i}|T_{a_1}, T_{a_2}, \dots, T_{a_{i-1}}, Z^N) - |{\mathscr A}| \nonumber\\
&\hspace{-0.5in} \geq |\cI | (1-O(N^2 2^{-N^{\beta}}) )- |{\mathscr A}|  \label{lemma_ineq3_wiretap} \\
&\hspace{-0.5in} \geq |\cI \backslash {\mathscr A}|- O(N^3 2^{-N^{\beta}}) \nonumber
   \end{align}
where \eqref{lemma_ineq3_wiretap} is due to \eqref{lemma_ineq2_wiretap}.
\remove{
Inequality \eqref{lemma_ineq2_wiretap} is valid for any subset of indices in $\cI$
including in particular the subset ${\mathscr A},$ say ${\mathscr A}=\{a_1,a_2,\dots,a_{|{\mathscr A}|}\},$ which implies
\begin{equation*}
H(T[\cI\backslash {\mathscr A}]| T[{\mathscr A}]\, Z^N)=\sum_{i=|{\mathscr A}|+1}^{|\cI|} H(T_{a_i}|T_{a_1} T_{a_2} \dots T_{a_{i-1}}\, Z^N) \ge 
|\cI\backslash {\mathscr A}| (1-2\delta_N).
\end{equation*}
}
This completes the proof of \eqref{lemma_ineq_wiretap}.
\end{proof}
%

%Lemma \ref{security_lemma_wiretap} is needed to establish the base for Lemma \ref{security_lemma_wiretap2}.
\begin{lemma}
Let  $M^j\triangleq (M_1,M_2,\dots,M_j)$ and $Z^j\triangleq(Z_1,Z_2,\dots,Z_j).$ 
For all $j=1,\dots, m$, we have
   \begin{equation*}
I(E_j; M^j ,Z^j)=O(jN^3 2^{-N^{\beta}}).
   \end{equation*}
\label{security_lemma_wiretap2}
\end{lemma}
\begin{proof} The proof is by induction on $j.$ The base case $j=1$ follows from Lemma \ref{security_lemma_wiretap}.
Now assume that the claim of the lemma is true for $j=k-1$ and write 
    \begin{align*}
I(E_k; M^k, Z^k) &= I(E_k; M_k, Z_k)\\&\quad+ I(E_k; M^{k-1}, Z^{k-1}| M_k, Z_k). 
  \end{align*}
Using the chaining structure shown in Figure \ref{figure3} we argue that the only
part of the transmission that connects block $k-1$ to block $k$ is given by $E_{k-1}.$
This implies that, conditional on $(M_k,Z_k)$, we have a Markov chain 
$M^{k-1}Z^{k-1}\to E_{k-1}\to E_k,$ so
  \begin{align*}
  &I(E_k; M^{k-1}, Z^{k-1}| M_k, Z_k) \\ &\quad\le I( E_{k-1}; M^{k-1}, Z^{k-1}| M_k, Z_k) 
  \\&\quad\le I( E_{k-1}; M^{k-1}, Z^{k-1}),
  \end{align*}
where the second inequality is due to the Markov chain $M^{k-1} Z^{k-1}\to E_{k-1}\to M_k Z_k.$
Therefore,
  \begin{align}
I(E_k; M^k, Z^k)&\leq I(E_k; M_k ,Z_k)+ I(E_{k-1}; M^{k-1}, Z^{k-1})  \nonumber\\
&\hspace{-0.2in}=O(N^3 2^{-N^{\beta}})+ O( (k-1) N^3 2^{-N^{\beta}}) \label{wiretap_secure_2} \\
&\hspace{-0.2in}= O( k N^3 2^{-N^{\beta}}). \nonumber
   \end{align}
   Here \eqref{wiretap_secure_2} is due to Lemma \ref{security_lemma_wiretap} and the induction hypothesis. This completes the proof.
\end{proof}

Now we are ready to complete the proof of the strong security condition. For this purpose, consider the following sequence
of inequalities:
\begin{align}
&I( M^m; Z^m) \leq I(M^m; Z^m, E_m)  \nonumber \\
&= I(M^{m-1}; Z^m, E_m)+ I(M_m; Z^m, E_m | M^{m-1}) \nonumber \\
&\leq I(M^{m-1}; Z^{m-1}, E_{m-1}) + I(M_m; Z^m, E_m | M^{m-1}) \label{wiretap_secure_4}   \\
&\leq I(M^{m-1}; Z^{m-1}, E_{m-1}) + I(M_m; Z^m E_m , M^{m-1})  \nonumber\\
& = I(M^{m-1}; Z^{m-1}, E_{m-1}) + I(M_m; E_m, Z_m) \nonumber\\ &\quad+I(M_m; M^{m-1}, Z^{m-1}| E_m, Z_m) \nonumber\\
& \leq  I(M^{m-1}; Z^{m-1}, E_{m-1}) + I(M_m; E_m, Z_m)\nonumber\\ &\quad+ I(E_{m-1}; M^{m-1}, Z^{m-1}) \label{wiretap_secure_5}   \\
& =  I(M^{m-1}; Z^{m-1}, E_{m-1}) + O( N^3 2^{-N^{\beta}})\nonumber\\&\quad+ O( (m-1) N^3 2^{-N^{\beta}}) \label{wiretap_secure_6} \\
&=  I(M^{m-1}; Z^{m-1} E_{m-1}) +O (m N^3 2^{-N^{\beta}}), \label{wiretap_secure_7}
\end{align}
where \eqref{wiretap_secure_4} and \eqref{wiretap_secure_5} are implied by the chaining
structure in Figure \ref{figure3}, and \eqref{wiretap_secure_6} is due to Lemmas \ref{security_lemma_wiretap} and
\ref{security_lemma_wiretap2}.  From \eqref{wiretap_secure_7}, we have
\begin{align*}
 I(M^m; Z^m, E_m) &\leq I(M^{m-1}; Z^{m-1}, E_{m-1}) \\ &\quad+O (m N^3 2^{-N^{\beta}}) 
\end{align*}
which implies
\begin{equation*}
I(M^m;Z^m)\leq I(M^m; Z^m E_m)= O(m^2 N^3 2^{-N^{\beta}}) .
\end{equation*}
Then, recalling that $m=N^{\alpha}, \alpha>0$ is sufficient to  
satisfy the rate constraint, we observe that $I(M^{m};\,Z^{m})= O(N^{3+2\alpha}\, 2^{-N^{\beta}})\to 0$
as $N\to\infty$, as required.
}

\remove{
The proof follows the argument in \cite{sasoglu}. Writing $E_j$ instead of $\cE(j),$ we have
   % \begin{equation*}
    \begin{align}
I(M^{m}; &Z^{m}) \le I(M^{m}E_m; Z^{m})   \notag\\
& =I(M^mE_m;Z_m)+I(M^mE_m;Z^{m-1}|Z_m)   \notag\\
& =I(M_{m} E_{m};\,Z_m)+I(M^mE_{m};Z^{m-1}|Z_m)  \notag\\
&\le I(M_mE_{m};Z_m)+I(M^mE_{m};Z^{m-1}|Z_m)+I(Z_m;Z^{m-1})+I(E_{m-1};Z^{m-1}|M^mE_{m}Z_m) \label{eq:st}\\
&=I(M_mE_{m};Z_m)+I(M^m(E_{m-1}E_{m})Z_m;Z^{m-1})   \notag\\
&=I(M_mE_{m};Z_m)+I(M^{m-1}E_{m-1}; Z^{m-1})   \notag\\
&\le \sum_{j=1}^m I(M_j E_j;Z_j)   \notag
   \end{align}
%   \end{equation*}
where the second equality relies on the Markov condition $M^{m-1}-M_mE_{m}-Z_m$ and the next-to-last line on the condition
 $M_mE_{m}Z_m-M^{m-1}E_{m-1}-Z^{m-1}.$
Since the message bits are entirely within the set of indices $\cH_{V|Z},$ the corresponding Bhattacharyya parameters are close to 1.
Invoking the relation $Z(X|Y)^2\le H(X|Y)$ \cite{Arikan2010}, we find for all $j$ that $I(M_j E_j;Z_j)\le 2 N2^{-N^\beta}.$ %which tends to $0$ as $N\to\infty.$
Hence, $I(M^{m}; Z^{m})\le 2 m N2^{-N^\beta}$, which implies that  
$\lim_{N\to\infty} I(M^{m};\,Z^{m})=0$ as required.
}

We conclude that a secrecy rate of $I(V;Y)-I(V;Z)$ is achievable for any $V$ such that $V\to X\to Y,Z$ holds.
Therefore, the secrecy capacity $C_s$ given by \eqref{eq:cs} is also achievable.
\end{proof}

\remove{\begin{remark} Introducing the probability $P_{M}(m)$ of transmitting the message $m\in{\mathscr M}$
and the conditional probabilities $P^N_{Y|X}(y^N|x^N),$ $P^N_{Z|X}(z^N|x^N)$ of observations for Receivers 1 and 2, we can write
the mutual information term \eqref{security} in more detail as follows:
$$
I(M;Z^N)= \sum_{\begin{substack}{m\in {\mathscr M}\\ x^n\in {\mathscr X}^N}\end{substack}} 
P_{M}(m) f(x^N|m) P^N_{Z|X}(z^N|x^N) \log\biggl[ 
\frac{\sum_{x^N\in {\mathscr X}^N}  f(x^N|m) P^N_{Z|X}(z^N|x^N)}
{ \sum\limits_{{m'\in {\mathscr M}, x^N\in {\mathscr X}^N}}  
P_{M}(m') f(x^N|m') P^N_{Z|X}(z^N|x^N) }\biggr].
$$
\end{remark}}

%%%%%
%%%%%Section Broadcast Channels
%%%%%
\section{Polar Coding for Broadcast Channel with Confidential Messages}\label{sect:bcc}
In this section we observe that ideas of the previous section together with some earlier works 
enable us to extend our code construction to a more general communication model introduced in \cite{csiszar2}.
\subsection{The Model}\label{BCC_intro}
Consider a pair of discrete memoryless channels with one transmitter $X$ and two receivers $Y,Z.$
As before, let $W_1:X\to Y$ and $W_2:X\to Z$ denote the channels and let
$\mathscr{X}$ and $\cY,\cZ$ denote the input alphabet and the output alphabets.
We assume that the system transmits three types of messages: 
\begin{enumerate}
\item[(i)\hspace*{.08in}] a message $s_1\in\cS_1$ from $X$ to $ Y$ for which there
are no secrecy requirements;
\item[(ii)~] 
a message $s_2\in\cS_2$ from $X$ to $Y$ which
is secret from $Z;$ 
\item[(iii)] a message $t\in\cT$ from $X$ to $Y$ and $Z$, called the ``common message''.
\end{enumerate}
Following \cite{csiszar2}, we call this communication scheme a {\em broadcast channel with confidential messages} (BCC).

As before, a block encoder for the BCC is a mapping $f:\mathscr{S}_1\times \mathscr{S}_2 \times
\mathscr{T}\to \mathscr{X}^N.$  A stochastic version
of the encoder is a probability matrix $f(x^N|s_1,s_2,t)$ with columns indexed by $x^N\in \cX^N$ and rows indexed by the triples
$(s_1,s_2,t).$ Given such a triple, the stochastic encoder 
samples from the conditional probability distribution on $\cX^N.$
In accordance with the problem statement, there are two decoders: The decoder of Receiver 1 is defined by a mapping
$\phi: {\mathscr Y}^N \to {\mathscr S}_1\times {\mathscr S}_2 \times{\mathscr T}$ and the decoder of Receiver 2 is
a mapping $\psi:{\mathscr Z}^N \to {\mathscr T}$.

Denote the rate of the common message $t$ by $R_0$, and denote the rates of the secret and non-secret messages to $Y$ by $R_s$ and $R_1,$
respectively. 
The analogs of Definitions~\ref{def:1.1},~\ref{def:1.2} in this case look as follows.
\begin{definition}
The encoder-decoder mappings $(f,\phi,\psi)$ give rise to \emph{$(N,\epsilon)$- transmission} over
the BCC if for every $s_1\in {\mathscr S}_1$, $s_2\in {\mathscr S}_2$, $t\in {\mathscr T}$, 
decoder $\phi$ outputs the transmitted triple $(s_1,s_2,t)$ and decoder $\psi$ outputs the message $t$ with probability greater than 
$1-\epsilon$, i.e.,
\begin{gather*}
\sum_{x^N\in {\mathscr X}^N} f(x^N|s_1,s_2,t) P_{Y^N|X^N} (\phi(y^N)=(s_1,s_2,t)|x^N)\\ \geq 1-\epsilon, \\
\sum_{x^N\in {\mathscr X}^N} f(x^N|s_1,s_2,t) P_{Z^N|X^N} (\psi(z^N)=t|x^N) \geq 1-\epsilon.
\end{gather*}
\end{definition}
\begin{definition}
$(R_1,R_s,R_0)$ is an achievable rate triple for the BCC if there exists a sequence
of message sets ${\mathscr S}_{1,N}$, ${\mathscr S}_{2,N}$, ${\mathscr T}_N$ and encoder-decoder triples
$(f_N, \phi_N, \psi_N)$ giving rise to $(N,\epsilon_N)$ transmission with $\epsilon_N\to 0$,
such that
   \begin{equation}\label{eq:ck}
\begin{aligned}
\lim_{N\to\infty} \frac{1}{N} \log |{\mathscr S}_{1,N}| &=R_1 \\
\lim_{N\to\infty} \frac{1}{N} \log |{\mathscr S}_{2,N}| &=R_s \\
\lim_{N\to\infty} \frac{1}{N} \log |{\mathscr T}_{N}| &=R_0 \\
\lim_{N\to\infty} I(S_{2,N}; Z^N)&=0.
\end{aligned}
    \end{equation}
where $S_{2,N}$ is the random variable that corresponds to the secret message.
\label{BCC}
\end{definition}
Note that our definition takes into account the formulation in \cite{csiszar} and is
slightly different from the one in \cite{csiszar2} (where we write $(R_1,R_s,R_0)$, \cite{csiszar2} has $(R_1+R_s,R_s,R_0)$~). 
%This is because in \cite{csiszar2}, 
%$R_1$ refers to the rate of the message that Receiver 1 gets separately including the secret and non-secret part, 
%not just the non-secret part. So, whenever $(R_1,R_s,R_0)$ is an achievable rate triple for us, 
%$(R_1+R_s,R_s,R_0)$ is an achievable rate triple in the formulation of \cite{csiszar2}.

\nc{\reg}{{\mathfrak R}}

The following theorem gives the achievable rate region for the triple $(R_1,R_s,R_0)$.
\begin{theorem}\label{thm:bcc}\cite{csiszar2},\cite[p.~414]{csiszar}
The capacity region of the BCC  consists of those triples of nonnegative numbers 
$(R_1,R_s,R_0)$ that satisfy, for some RVs $U\to V \to X \to Y,Z$ with $P_{Y|X}=W_1$
and $P_{Z|X}=W_2$, the inequalities
   \begin{align}
R_0 &\leq \min [I(U;Y),I(U;Z)] \label{eq:r0}\\
R_s &\leq I(V;Y|U)-I(V;Z|U) \label{eq:rs}\\
R_0+R_1+R_s &\leq I(V;Y|U)+ \min [I(U;Y),I(U;Z)].\label{eq:r}
   \end{align}
Moreover, it may be assumed that $V=(U,V')$ and the range sizes of $U$ and $V'$ are at most
$|X|+3$ and $|X|+1$.
\label{main}
\end{theorem}
One can define the secrecy capacity $C_s$ in terms of the capacity region of the BCC, and then recover
Theorem \ref{thm:dg} as a particular case of Theorem~\ref{main}.
%\begin{definition} The secrecy capacity is the maximum $R_s$ for which $C_s=\sup R_s: (0,R_s,0)\in \reg.$
%\end{definition}
%Theorem \ref{thm:dg} can be considered as a special case of Theorem \ref{main}.
%As a corollary to Theorem \ref{main}, one can derive the expression for $C_s$ 
%as given by Theorem \ref{thm:dg}.

\subsection{Polar Coding for the Csisz{\'a}r-K{\" o}rner Region}
\label{csiszar-korner}
In this section, we aim to show that the capacity region of the BCC can be achieved
using polar codes. 
%Our proof consists of three steps and depends on the RVs 
%$U\to V \to X \to Y,Z$ with $P_{Y|X}=W_1$ and $P_{Z|X}=W_2$ . 
In the first two steps we design a scheme that achieves the rate pairs $(R_0,R_s)$ in \eqref{eq:r0}-\eqref{eq:rs},
and in the last step we show 
that for any such pair $(R_0,R_s)$ any rate value
\begin{equation}
R_1\leq I(V;Y|U)+ \min [I(U;Y),I(U;Z)]-R_0-R_s\label{R_1}
\end{equation}
is also achievable. Finally, the security condition in \eqref{eq:ck} will be shown in Proposition \ref{prop:sc} below.
%From Theorem \ref{main}, we see that these three steps will be sufficient to complete the proof.

The overall encoding scheme is stochastic and assumes some fixed joint distribution of the RVs 
$U,V,X,Y,Z$ such that the constraints of Theorem \ref{thm:bcc} are satisfied. Since the results below are valid for
any such distribution, this will enable us to claim achievability of the rate region in this theorem.
The encoder is formed of two stages performed in succession. At the outcome of
the first stage, which deals with the common message $s_2,$ the encoder 
computes a sequence of $m$ blocks of $N$ bits denoted below by $q^{N}(j), j=1,\dots,m.$
These blocks are used in the second stage to construct the data encoding that is going to be sent to both receivers. 
Namely, it will be seen that the transformed blocks $u^{N}=q^{N}G_N$ can at the same time encode the common
message to both receivers and also encode side information for Receiver 1 to ensure reliable
transmission of the confidential message. The actual sequences to be transmitted are computed in the second
stage based on the sequences $u^N(j).$ This is done by first constructing sequences $t^N(j)$ using the
ideas developed in Sec.~\ref{sect:wt} and by using a stochastic mapping of these sequences on the condewords
$x^N(j).$ Upon transmitting, these codewords are received by Receiver 1 as $y^N(j)$ and by Receiver 2 as
$z^N(j).$ We will argue that the receivers can independently perform decoding procedures that recover the three
desired types of messages reliably (and when appropriate, also securely).

\vspace*{.05in}\subsubsection{The common-message encoding} \label{sect:common}
The proof of
the fact that any $R_0$ satisfying \eqref{eq:r0} is achievable follows from the polar
coding scheme for the superposition region given in \cite{Mondelli14}. 
Given the RVs $U\to V \to X \to Y,Z$ with $P_{Y|X}=W_1$ and $P_{Z|X}=W_2$,
let $U^{N},V^{N}, X^{N}, Y^{N} ,Z^{N}$ be $N$ independent repetitions of the RVs
$U,V,X,Y,Z$. Set 
  \begin{equation}\label{eq:Q}
  Q^{N}=U^{N}G_N,
  \end{equation}
   where $G_N$ is Ar{\i}kan's transform. 
 As before, lowercase letters denote realizations of these RVs.

Define the sets ${\mathscr H}_U$, ${\mathscr L}_{U|Y}$, ${\mathscr L}_{U|Z}$ as follows:
\begin{align*}
{\mathscr H}_U&=\{i\in [N]: Z(Q_i|Q^{i-1}) \geq 1-\delta_N\}  \\
{\mathscr L}_{U|Y}&= \{i\in [N]: Z(Q_i|Q^{i-1},Y^{N}) \leq \delta_N\} \\
{\mathscr L}_{U|Z}&= \{i\in [N]: Z(Q_i|Q^{i-1},Z^{N}) \leq \delta_N\}.
\end{align*}
The cardinalities of these sets, normalized by $N$, approach respectively $H(U), 1-H(U|Y), 1-H(U|Z)$ as $N\to\infty.$
\remove{
Then, set $T^{1:n}=V^{1:n} G_n$. 
By thinking of $U$ as side information on $V$,
we define the sets ${\mathscr H}_{V|U}$, ${\mathscr L}_{V|U}$, ${\mathscr H}_{V|U,Y}$
and ${\mathscr L}_{V|U,Y}$ as
\begin{align*}
{\mathscr H}_{V|U}&=  \{i\in [n]: Z(T^i|T^{1:i-1},U^{1:n}) \geq 1-\delta_n\}, \\
{\mathscr L}_{V|U}&=  \{i\in [n]: Z(T^i|T^{1:i-1},U^{1:n}) \leq \delta_n\}, \\
{\mathscr H}_{V|U,Y}&= \{i\in [n]: Z(T^i|T^{1:i-1},U^{1:n},Y^{1:n}) \geq 1-\delta_n\}, \\
{\mathscr L}_{V|U,Y}&= \{i\in [n]: Z(T^i|T^{1:i-1},U^{1:n},Y^{1:n}) \leq \delta_n\}
\end{align*}
which satisfy
\begin{align*}
\lim_{n\to\infty} \frac{1}{n} |{\mathscr H}_{V|U}|&= H(V|U),\\
\lim_{n\to\infty} \frac{1}{n} |{\mathscr L}_{V|U}|&= 1- H(V|U),\\
\lim_{n\to\infty} \frac{1}{n} |{\mathscr H}_{V|U,Y}|&= H(V|U,Y),\\
\lim_{n\to\infty} \frac{1}{n} |{\mathscr L}_{V|U,Y}|&= 1-H(V|U,Y).
\end{align*}
}

Now observe that for Receiver 1 to recover
$q^{N}$ correctly, the indices of the information bits should be a subset of ${\mathscr I}^{(1)}_u={\mathscr H}_U \cap {\mathscr L}_{U|Y}$. Similarly, for Receiver 2 to recover the sequence $q^{N}$ correctly, the information
bits should be placed only in those positions of $q^{N}$ that are indexed by the set 
${\mathscr I}^{(2)}_u= {\mathscr H}_U \cap {\mathscr L}_{U|Z}$.
Therefore, choosing the indices of information bits as $ {\mathscr I}_u= {\mathscr I}^{(1)}_u\cap {\mathscr I}^{(2)}_u$  ensures that the message embedded 
into $q^{N}$ will be decoded correctly by both receivers. 
In this case, the rate of the common message is 
$R_0=|{\mathscr I}^{(1)}_u\cap {\mathscr I}^{(2)}_u|/N$. Given that
\begin{align*}
\lim_{N\to\infty} \frac{1}{N} |{\mathscr I}^{(1)}_u|&=I(U;Y)\\
\lim_{N\to\infty} \frac{1}{N} |{\mathscr I}^{(2)}_u|&=I(U;Z)
\end{align*}
we conclude the common message rate $R_0$ attains the value $\min[I(U;Y);I(U;Z)]$ only if either ${\mathscr I}^{(2)}_u \subseteq {\mathscr I}^{(1)}_u$
or ${\mathscr I}^{(1)}_u \subseteq {\mathscr I}^{(2)}_u$ holds starting from some $N$.
However, generally this does not have to be the case. To overcome this problem, 
\cite{Mondelli14} proposed the following coding scheme. Define
the sets
\begin{align*}
{\mathscr D}^{(1)}&= {\mathscr I}^{(1)}_u\backslash {\mathscr I}^{(2)}_u\\
{\mathscr D}^{(2)}&= {\mathscr I}^{(2)}_u\backslash {\mathscr I}^{(1)}_u.
\end{align*}

Without loss of generality, assume that $I(U;Y)\leq I(U;Z)$, which implies that $|{\mathscr D}^{(2)}| \geq |{\mathscr D}^{(1)}|$ 
starting with some $N$. To describe the encoding procedure, consider $m$ blocks of $N$ coordinates each. 
In block 1, we use the positions indexed by ${\mathscr D}^{(1)}$ to store message bits and
assign the bits indexed by ${\mathscr D}^{(2)}$ to some fixed values that are available to Receiver 1.
In Block $j, j=2,\dots,m-1,$ we again use the positions indexed by ${\mathscr D}^{(1)}$ 
to store message bits and copy the part of Block $j-1$ indexed by 
the coordinates in ${\mathscr D}^{(1)}$ into the positions indexed by a subset of coordinates
${\mathscr E}^{(2)}\subset {\mathscr D}^{(2)}$ in block $j$ (thus $|{\mathscr E}^{(2)}|=|\cD^{(1)}|$).
Fill the remaining $|{\mathscr D}^{(2)}\backslash \cE^{(2)}|$ bits in each block $j\in \{2,\dots,m-1\}$ with 
random and independent bits and communicate them to Receiver 1. 
These bits can be the same for each block as long as they are independent and uniform within the same block,
so this part of the scheme has negligible impact on the overall rate.
In the final block $m$, we assign the bits indexed by ${\mathscr D}^{(1)}$ to 
some fixed values that are available to Receiver 2 and copy the bits in ${\mathscr D}^{(1)}(m-1)$ 
to the positions in ${\mathscr E}^{(2)}(m).$ 
The remaining $|{\mathscr D}^{(2)}\backslash \cE^{(2)}|$ coordinates in block $m$ are filled with random bit values.
{The bits in ${\mathscr H}_U \cap (\cI_u^{(1)} \cup \cI_u^{(2)})^c$ and ${\mathscr H}_U^c$ are chosen
based on deterministic rules $\lambda_i$ similarly to Sec.~\ref{sect:wt}. These bits are the same for each block and are shared with both
receivers.}
The block diagram of the described coding scheme is shown by Figure \ref{figure4}.
\begin{figure*}[t]
\centering
\includegraphics[scale=0.7]{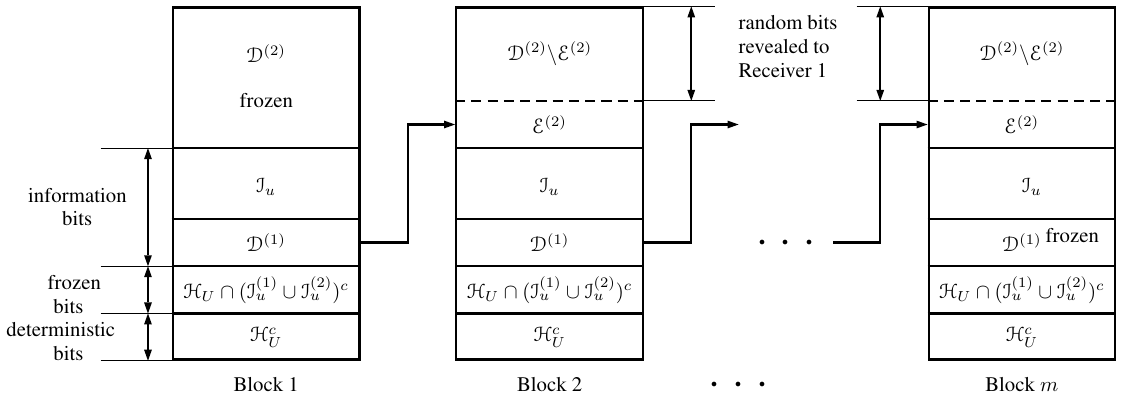}
\caption{Encoding for the common-message case: The structure of the blocks $q^N(j),j=1,\dots,m$}\label{figure4}
\end{figure*}

The encoding stage described passes the sequences $u^N(j),j=1,\dots,m$ to the second stage which is responsible for actual communication.
Upon observing the channel outputs, both receivers perform decoding, which will be described below in 
Sec.~\ref{subsec:decoding}.
\subsubsection{The secret-message encoding}
\label{secondsubsec}

In this section we describe the construction of sequences $x^N$ that are sent by transmitter $X.$ 
The construction relies on the sequences $u^N(j)$ constructed by $X$ in the first stage.
These sequences can be thought of as side information that enables Receiver 1 to reconstruct the
secret message.

The transmission scheme we propose to achieve the rate $R_s$ that satisfies \eqref{eq:rs},
is very similar to the scheme described for the wiretap channel problem in
Sec. \ref{sect:wt}. Our solution consists of choosing the indices of information
bits and random bits appropriately and using a chaining scheme quite similar to the one shown in Figure \ref{figure3}.
 
Let  $U^{N},V^{N}, X^{N}, Y^{N} ,Z^{N}$ be as defined in Sec. \ref{sect:common},
and let $T^{N}=V^{N} G_N$. Viewing $U$ as side information about $V$,
we define the sets
\begin{align*}
{\mathscr H}_{V|U}&=  \{i\in [N]: Z(T_i|T^{i-1},U^{N}) \geq 1-\delta_N\} \\
{\mathscr L}_{V|U,Y}&= \{i\in [N]: Z(T_i|T^{i-1},U^{N},Y^{N}) \leq \delta_N\}  \\
{\mathscr H}_{V|U,Z}&= \{i\in [N]: Z(T_i|T^{i-1},U^{N},Z^{N}) \geq 1-\delta_N\} 
\end{align*}
whose cardinalities, normalized by $N,$ approach respectively the values $H(V|U),1-H(V|U,Y),H(V|U,Z)$ as $N\to\infty.$

%The polar coding scheme described in Section \ref{sect:common} makes it possible for Receiver 1 to have a side information $U^{N}$. Hence, if the information

\vspace*{.1in} The intuition behind the construction presented below can be described as follows.
First, note that the coordinates of $T^{N}$ indexed by 
  \begin{equation}\label{eq:J1}
   {\mathscr J}^{(1)}={\mathscr H}_{V|U}\cap {\mathscr L}_{V|U,Y}
   \end{equation}
     can be decoded by Receiver 1, and so they can be used to send more data in addition to the
common message. Of these bits, the part indexed
by $
({\mathscr H}_{V|U}\cap {\mathscr L}_{V|U,Y})\backslash ({\mathscr H}_{V|U}\cap {\mathscr H}^c_{V|U,Z}) $
can be used to transmit the confidential message. Then, given that 
$
\frac{1}{N} |{\mathscr H}_{V|U}\cap {\mathscr L}_{V|U,Y}|\to I(V;Y|U)$
and $\frac{1}{N} |{\mathscr H}_{V|U}\cap {\mathscr H}^c_{V|U,Z}|\to I(V;Z|U),
$
we obtain
\begin{gather*}
\lim_{N\to\infty}\frac{1}{N} |({\mathscr H}_{V|U}\cap {\mathscr L}_{V|U,Y})\backslash ({\mathscr H}_{V|U}\cap {\mathscr H}^c_{V|U,Z})| 
\\ \geq I(V;Y|U)-I(V;Z|U).
\end{gather*}
This implies that the proposed scheme transmits the secret message at rates arbitrarily close to the rate
given by \eqref{eq:rs}
 (provided that it also satisfies the strong security condition).

%\textcolor{blue}{In \eqref{eq:J1} I changed notation from $\cI^{(1)}$ to $\cJ^{(1)}$ to avoid the conflict with
%$\cI^{(1)}_u$ on the previous page. Please confirm.}

\vspace*{.1in}Building on this observation, we proceed to describe the coding scheme, adding some details
that make the secrecy part work. Define the sets\footnote{We again use the same notation as in \eqref{eq:p}; since the earlier notation
is used only in Sec. \ref{sect:wt}, this should not cause confusion.}
   \begin{equation}
   \begin{aligned}
         \label{partitionBCC}
{{\mathscr I}}&= {\mathscr H}_{V|U}\cap {\mathscr L}_{V|U,Y}\cap {\mathscr H}_{V|U,Z} \\
{{\mathscr B}}&= {\mathscr H}_{V|U}\cap {\mathscr L}^c_{V|U,Y} \cap {\mathscr H}_{V|U,Z} \\
{{\mathscr R}}_1&={\mathscr H}_{V|U}\cap {\mathscr L}_{V|U,Y}\cap {\mathscr H}^c_{V|U,Z} \\
{{\mathscr R}}_2&= {\mathscr H}_{V|U}\cap {\mathscr L}^c_{V|U,Y}\cap {\mathscr H}^c_{V|U,Z} \\
{{\mathscr D}}&= {\mathscr H}^c_{V|U}.
   \end{aligned}
   \end{equation}
Note that the sets ${{\mathscr I}},{{\mathscr R}}_1,{{\mathscr R}}_2,{{\mathscr B}},{{\mathscr D}}$ partition $[N]$.
This partition is basically the same as in \eqref{eq:p} (see also Figures 3 and \ref{fig:partition}) except for the fact that the high- and low-entropy subsets rely on entropy quantities that are additionally conditioned on $U$.
\remove{\begin{figure}[h]
\begin{picture}(200,200)(80,0)
\includegraphics[scale=1.1]{partition.eps}
\put(-300,3){${{\mathscr R}}_1$}
\put(-220,-10){${{\mathscr I}}$}
\put(-200,23){${{\mathscr E}}$}
\put(-160,3){${{\mathscr B}}$}
\put(-115,3){${{\mathscr R}}_2$}
\put(-50,3){${{\mathscr D}}$}
\put(-380,40){$V\to U,Y$}
\put(-380,105){$V\to U,Z$}
\put(-370,178){$V\to U$}
\put(-170,155){$N$}
\put(-220,192){$\cH_{V|U}$}
\put(-200,124){$\cH_{V|U,Z}$}
\put(-270,58){$\cL_{V|U,Y}$}
\end{picture}
\vspace*{.1in}
\begin{caption} {Partition of coordinates for transmission over the broadcast channel.
The top block of $N$ coordinates refers to high-entropy coordinates for
the distribution $P_{V|U}$. Similarly, the middle block of $N$ coordinates corresponds to high-entropy 
coordinates for the distribution $P_{V|U,Z}$. The low-entropy coordinates for the distribution $P_{V|U,Y}$
is represented by the bottom block of $N$ coordinates. The partition of $[N]$ is done based on these three blocks.}\label{partition2}
\end{caption}
\end{figure}
}

The transmission scheme that we propose is formed of multiple blocks joined in clusters of $m$ blocks.
Similarly to the wiretap coding scheme, 
there is a seed block shared between the transmitter and Receiver 1. 
The seed block consists of $|{{\mathscr R}}_2|$ random bits.
Even if the set ${{\mathscr R}}_2$ constitutes a nonvanishing proportion of $[N]$, 
the rate of the seed $|{{\mathscr R}}_2|/ mN$ can be
made arbitrarily small by choosing $m$ sufficiently large. (For example, one can
set $m=N^{\alpha}$ for some $\alpha>0$, and let $N$ tend to infinity.)

%Now, we are ready to describe the encoder and decoder to achieve $R_s \leq I(V;Y|U)-I(V;Z|U)$.
\vspace*{.1in}
The encoding procedure is as follows.
Our aim is to construct $m$ blocks $t^N(j)$ which will be used to form the transmitted sequences 
$x^N(j), j=1,\dots,m.$
Apart from the seed block, all the other blocks $t^N$ contain a group of almost deterministic bits, 
denoted by ${{\mathscr D}}$ in \eqref{partitionBCC}. For block $j, j=1,\dots,m$, the values of these bits are assigned according to 
a deterministic rule $\lambda_i$ (similarly to the earlier appearence of this mapping, see, e.g., Sec.~\ref{sect:wt}), namely
    $$
      t_i=\lambda_i  (t^{i-1}, u^N(j))
    $$
 for all $i\in{{\mathscr D}}$.
The set of bits ${{\mathscr B}}$ in each of the $m$ blocks is chosen similarly based on
 a family of deterministic rules $\{\lambda_i, i\in {{\mathscr B}}\}.$
 {These rules are chosen the same for each block, and are shared with Receiver 1 (cf. also the remark on this in Sec.~\ref{sect:wt}).} 
\remove{Even if the rate $(|\cB|+|\cD|)/mN$ can be made arbitrarily small similarly to $|\cR_2|/mN$, it may be
the case that such a rule sharing requires a positive rate of secure communication. As a remedy to this
problem, the bits $t_i,i\in\cB$ can be chosen randomly and the same for each block, and shared with
Receiver 1 secretly. It is still possible to maintain secrecy by using this solution, as the secrecy proof given here
and in \cite{chou14} indicates} 

The remaining subsets of coordinates in are filled as follows. For block 1, the bits in the set ${{\mathscr R}}_2$ are assigned 
the values of the bits of the seed block, while for block $j$, $j=2,\dots,m$ these bits are set to be equal to the bits 
in ${{\mathscr E}}$ of block $j-1.$ 
Here, $\cE$ is a subset of $\cI$ having the same size as $\cR_2$.
The messages are stored in the bits indexed by ${{\mathscr I}}\backslash{{\mathscr E}}$ which are good for Receiver 1 and
 contained in the bad
(high-entropy) set of Receiver 2.
The indices in the subset ${\mathscr E}$ are still good for Receiver 1 but bad for Receiver 2. 
Nevertheless, they are filled with random bits that are used for decoding by Receiver 1 in the same way as was done in Sec.~\ref{sect:wt}.
Finally, the bits in ${{\mathscr R}}_1$ are assigned randomly and uniformly for each of the $m$ blocks. We once again refer to
Fig.~\ref{fig:partition} which illustrates the described processing.
Once the blocks $t^N(j), j=1,\dots,m$ are formed, we compute $m$ sequences $v^N(j) =t^N(j)G_N$ by using the polarizing transform. 

Finally, the codewords to be sent by the transmitter are computed as follows. The codeword $x^N(j), j=1,\dots,m$ is sampled from
$\cX^N$ according to the distribution 
 $P_{X^{N}|V^{N}}(x^{N}|v^{N})=\prod_{i=1}^N P_{X|V}(x_i|v_i)$, where
$P_{X|V}$ is the conditional distribution induced by the joint distribution of the RVs $V$ and $X$.

\subsubsection{Decoding of the common message and the secret message}\label{subsec:decoding} 
%\vspace*{.1in}
 Assume that the transmitted sequence $x^N$ is received as $y^N$ by Receiver 1 and as $z^N$ by Receiver 2.
Importantly, by our construction these sequences follow the conditional distributions $P_{Y|X}$ and $P_{Z|X}$ given by the channels $W_1$ and $W_2.$ 
We describe the decoding procedures by Receivers 1 and 2. Initially they perform similar operations 
aimed at recovering the common message.
Once this is accomplished, Receiver 1 performs additional decoding to recover the secret message. 

We begin with the common-message part. In accordance with \eqref{eq:Q}, Receivers 1 and 2 decode
the blocks $q^N(j), j=1,\dots,m$ relying on a iterative procedure.
As the construction suggests, Receiver 1 decodes in the forward direction, starting with block $1$ and 
ending with block $m$, and Receiver 2 decodes backwards, starting with block $m$ and ending with block $1$. 
Let 
   $$
   \begin{aligned}
      D(j)&=\{q_i, i\in {\mathscr D}^{(1)}(j) \}, j=1,\dots,m-1\\
      C(j)&=\{q_i, i\in {\mathscr E}^{(2)}(j) \}, j=2,\dots,m
      \end{aligned}
   $$
 denote the subblocks ${\mathscr D},{\mathscr E}$ of the corresponding blocks (see Fig.~\ref{figure4}).

The processing by Receiver 1 is as follows.
For block $1$, it computes 
\begin{equation}\label{eq:r1-1}
\hat{q}_i= \begin{cases}
q_i \hspace*{.85in}\text{if} \,\, i\in {\mathscr D}^{(2)}\\
\argmax_{q\in \{0,1\}} P_{Q_i|Q^{i-1}Y^N}(q|\hat{q}^{i-1}, y^N), \\ 
\hspace*{1in}\text{if}\,\, i\in {\mathscr I}_u\cup {\mathscr D}^{(1)} \\
%\argmax_{q\in \{0,1\}} P_{Q_i|Q^{i-1} }(q|\hat{q}^{i-1}), &\text{otherwise}.
 \lambda_i (\hat{q}^{i-1}) \hspace*{.5in}\text{otherwise}.
\end{cases}
\end{equation}
For the remaining blocks $j=2,\dots,m$ Receiver 1 computes the vector $(\hat q_i(j), i=1,\dots,N)$ as follows:
   \begin{equation}\label{eq:r1-2}
\hat{q}_i(j)= \begin{cases}
D_i(j-1) \hspace*{.4in}\text{if} \,\, i\in {\mathscr E}^{(2)}\\
q_i(j) \hspace*{.7in} \text{if} \,\, i\in {\mathscr D}^{(2)}\backslash {\mathscr E}^{(2)}  \\
\argmax_{q\in \{0,1\}} P_{Q_i|Q^{i-1}Y^N}(q|\hat{q}^{i-1}(j), y^N(j)),\\ \hspace*{1.0in} \text{if}\,\, i\in {\mathscr I}_u\cup {\mathscr D}^{(1)} \\
%\argmax_{q\in \{0,1\}} P_{Q_i|Q^{i-1} }(q|\hat{q}^{i-1}(j)), &\text{otherwise}.
 \lambda_i (\hat{q}^{i-1}(j)) \hspace*{.35in} \text{otherwise}.
\end{cases}
   \end{equation}
The processing by Receiver 2 is quite similar except that it starts with block $m$ and advances ``backwards''
for $j=m-1,m-2,\dots,1.$ For block $m$ the rule is as follows:
   \begin{equation}\label{eq:r2-1}
\hat{q}_i= \begin{cases}
q_i \hspace*{.95in} \text{if} \,\, i\in {\mathscr D}^{(1)}\\
\argmax_{q\in \{0,1\}} P_{Q_i|Q^{i-1}Z^N}(q|\hat{q}^{i-1}, z^N), \\ \hspace*{1.1in} \text{if}\,\, i\in {\mathscr I}_u\cup {\mathscr D}^{(2)} \\
%\argmax_{q\in \{0,1\}} P_{Q_i|Q^{i-1} }(q|\hat{q}^{i-1}), &\text{otherwise}
 \lambda_i (\hat{q}^{i-1}) \hspace*{.6in} \text{otherwise}.
\end{cases}
\end{equation}
For blocks $j=m-1,m-2,\dots,1$, Receiver 2 computes its estimates of the vector $q^N(j)$ as follows:
\begin{equation}\label{eq:r2-2}
\hat{q}_i(j)= \begin{cases}
C_i(j+1) \hspace*{.5in} \text{if} \,\, i\in {\mathscr D}^{(1)}\\
\argmax_{q\in \{0,1\}} P_{Q_i|Q^{i-1}Z^N}(q|\hat{q}^{i-1}(j), z^N(j)), \\ \hspace*{1.1in}\text{if}\,\, i\in {\mathscr I}_u\cup {\mathscr D}^{(2)} \\
%\argmax_{q\in \{0,1\}} P_{Q_i|Q^{i-1} }(q|\hat{q}^{i-1}(j)), &\text{otherwise}
 \lambda_i (\hat{q}^{i-1}(j))  \hspace*{.4in}\text{otherwise}.
\end{cases}
\end{equation}
(in \eqref{eq:r1-1}-\eqref{eq:r2-2} the notation is somewhat abbreviated to keep the formulas compact, e.g., 
no reference is made to the index of the receiver, and the block index $j$ is sometimes omitted).

\vspace*{.1in} The processing described above in \eqref{eq:r1-1}-\eqref{eq:r1-2}
yields the sequences $\hat u^N(j)=q^N(j)G_N, j=1,\dots,m$ which are used by Receiver 1 to recover the
secret messages.
Denote by ${\cE}(0)$ the message sequence encoded in
the seed block, and let ${\cE}(j), j=1,\dots,m$ be the subblock of block $j$ indexed by the set $\mathscr E.$
%i.e.,  $E_i(j)=\{t_i, i\in {{\mathscr E(j)}}\}.$  
%
For $j=1,\dots,m$, the decoding rule is as follows:
     \begin{align}\label{eq:dec1}
    &\hat{t}_i(j)=\nonumber\\ 
    &\begin{cases}
    \lambda_i (\hat{t}^{i-1}(j), \hat{u}^N(j)) \hspace*{.4in} \text{if}\,\, i\in {{\mathscr D}}\cup {{\mathscr B}} \\
\argmax_{t\in \{0,1\}} P_{T_i|T^{i-1} U^{N}Y^{N}}(t|\hat{t}^{i-1}(j), \hat{u}^N(j), y^{N}(j)), \\ \hspace*{1.47in}\text{if}\,\, i\in 
{{\mathscr I}}\cup{{\mathscr R}}_1\\
{\cE}_i(j-1) \hspace*{0.92in} \text{if} \,\, i\in {{\mathscr R}}_2.
                  \end{cases}
      \end{align}
where  $P_{T_i|T^{i-1}U^NY^{N}}$ is the conditional distribution induced by the
joint distribution of the RVs  $U^N$, $V^{N}$ and $Y^{N}$ (the notation is again abbreviated similarly to \eqref{eq:dec}). 

%\textcolor{blue}{I do not see the need for the new notation $E(j).$ I suggest a notational adjustment similar to the one made
%on p.10, Eq.~\eqref{eq:dec}, e.g., $E(j)\to\cE(j)$ starting 3 lines above Eq. (31) through the end of Section 5.2.3. Note also that $E_m$ in \eqref{ineq1} was not really defined; I assumed something and %added a few words. }

\subsubsection{Achievability of the rate region \eqref{eq:r0}-\eqref{eq:r}}\label{sect:achievability}

The rate of the common message achieved by the construction in Sec.~\ref{sect:common} is equal to
\begin{equation*}
R_0=\frac{1}{mN}\left[m |{\mathscr I}_u|+ (m-1) |{\mathscr D}^{(1)}|\right].
\end{equation*}
As $N$ increases, we obtain
\begin{equation*}
\frac{m-1}{m} I(U;Y)\leq R_0 \leq I(U;Y).
\end{equation*}
For sufficiently large $m$ this quantity is arbitrarily close to the
common-message rate value given in \eqref{eq:r0}. Note that
we have assumed that $I(U;Z)\geq I(U;Y);$ to handle the opposite case is suffices
to interchange the roles of the pieces ${\mathscr D}^{(1)}$ and ${\mathscr D}^{(2)}$ in the common-message
encoding and decoding procedures.

As shown in \cite{Mondelli14}, both receivers can decode the common message correctly with probability of
error at most $mN2^{-N^{\beta}}, \beta\in (0,1/2)$. This follows because in this stage, Receivers 1 and 2
aim only at decoding the bits corresponding to the index sets 
${\mathscr D}^{(1)}\cup {\mathscr I}_u= {\cH}_U\cap {\cL}_{U|Y}$ and
${\mathscr D}^{(2)}\cup {\mathscr I}_u= {\cH}_U\cap {\cL}_{U|Z}$, respectively. That these bits can be recovered
in a sucessive decoding procedure follows from the basic results on polar codes \cite{arikan2009,ari09a} and \cite{Honda13}.
%\textcolor{blue}{The last statement is inaccurate because the roles of the two receivers are not symmetric: the
%secret message goes to R1, as does the ``additional message''. Presumably, the interchange applies only to the common-message part. Please make more precise, pointing to the parts of the procedure %that need to be interchanged.}
%---------------
%\textcolor{blue}{Q: We need to argue that the decoded bits are reliable. This must be done by making an explicit statement, either with a proof or with a reference. At this point reliability is not mentioned %at all. }
%\textcolor{blue}{We also must prove that the fact that the same random bits are
%used in each block does not (somehow) compromise security: e.g., is it not possible to learn these bits gradually from many blocks and
%to use this at $Z$ to decode the secure message?}
%---------------

\vspace*{.1in} For the secret-message part of the communication the properties of the scheme are characterized by the following proposition.
\begin{proposition}\label{prop:sc}
For any $\gamma>0$, $\epsilon>0$ and $N\to\infty$ it is possible to choose $m$ so that the transmission scheme described 
above attains a secrecy rate $R_s$ such that $R_s\geq I(V;Y|U)-I(V;Z|U)-\gamma$ %\eqref{eq:cs}
and the information leaked to Receiver 2 satisfies the strong secrecy condition in Definition \ref{BCC}.
\end{proposition}
\begin{proof}
Assume that $X,V,Y,Z$ are as given by Theorem \ref{main}. The rate of proposed coding scheme is
\begin{equation*}
\frac{m (|{{\mathscr I}}|-|{{\mathscr E}}|)}{m N}= \frac{ (|{\mathscr H}_{V|U}\cap {\mathscr L}_{V|U,Y}|-|{\mathscr H}_{V|U}\cap {\mathscr H}^c_{V|U,Z}|)}{ N}
\end{equation*}
which converges to $I(V;Y|U)-I(V;Z|U)$ as $N$ goes to infinity. 
Recalling Theorem \ref{main}, we note that this is the target value of the rate $R_s$ for given $U,V$ and $X$ satisfying $U\to V\to X \to Y,Z$ and $P_{Y|X}=W_1$, $P_{Z|X}=W_2$.

Introduce the following RVs: Let $M^{m}=(M_1,M_2,\dots,M_m)$ be a sequence of message bits
 $\{t_i, i\in {{\mathscr I}} \backslash {{\mathscr E}}\}$ sent in blocks $1,\dots,m,$ and let 
$Z^m= (Z^N(1),\dots,Z^N(m))$ be the RVs that represent the random observations of Receiver 2 upon
transmitting $m$ blocks $X^N(j).$  Further, let $E_j$ correspond to the bits contained in the subset $\cE (j)$
for all $j=1,\dots,m.$

{\em Reliability:} The proof of reliable decoding by Receiver 1 follows from the results of
\cite{Honda13} and is very similar to Sec. \ref{sect:wt}. Let $P_{U^NV^NX^NY^NZ^N}(u^N,v^N,x^N,y^N,z^N)\triangleq
\prod_{i=1}^N P_{U,V,X,Y,Z} (u_i,v_i,x_i,y_i,z_i),$
where $P_{UVXYZ}$ is the joint distribution of the RVs $U,V,X,Y,Z$ appearing in
Theorem \ref{thm:bcc}. Let $\tilde P_{U^N,V^N,X^N,Y^N,Z^N}$ be the empirical
distribution induced by the transmission scheme if the deterministic rules $\lambda_i$
in both stages of the encoding are replaced by random assignments such that, for $a=0,1,$
  \begin{align*}
 \Pr(q_i=a)=P_{Q_i|Q^{i-1}}(a|q^{i-1})\\
  \Pr(t_i=a)=P_{T_i|T^{i-1}}(a|t^{i-1})
%P_{\Lambda_i}[\Lambda_i (q_1^{i-1})=1]=P_{Q_i|Q^{i-1}}(1|q^{i-1})\\
%P_{\Lambda_i}[\Lambda_i (t_1^{i-1})=1]=P_{T_i|T^{i-1}}(1|t^{i-1})
\end{align*}
holds for the first and second stages, respectively. Here $Q^N$ is the RV defined in \eqref{eq:Q}.  Then, from the proof of Lemma 1 in \cite{Honda13},
it follows that
\begin{align*}
\|P_{U^N}-\tilde P_{U^N}\| &\leq N 2^{-N^{\beta}} \\
\|P_{V^N|U^N}(.|u^N)-\tilde P_{V^N|U^N}(.|u^N)\| &\leq N 2^{-N^{\beta}},
\end{align*}
the second estimate for every $u^N.$
Hence, we conclude that
  \begin{align*}
\|P_{U^NV^N}-\tilde P_{U^NV^N}\| \leq 2N 2^{-N^{\beta}} \\
\|P_{U^NV^NX^NY^NZ^N}-\tilde P_{U^NV^NX^NY^NZ^N}\| \leq 2N 2^{-N^{\beta}} 
 \end{align*}
holds for all the $m$ blocks of transmission. Moreover, since the message bits are entirely 
contained in the set of ${\mathscr L}_{V|U,Y}$, the successive decoding procedure \eqref{eq:dec1}
has the probability of error bounded by
\begin{align*}
&|P_{U^NV^NX^NY^N}-\tilde P_{U^NV^NX^NY^N}\| \\ &\quad+ \sum_{i\in\cI} Z(T_i|T^{i-1},U^N,Y^N) \leq 3N 2^{-N^{\beta}}
\end{align*}
for each individual block. Thus, we observe that there exists a family of deterministic rules $\lambda_i$ for each
stage such that the overall error probability is at most 
$3mN 2^{-N^{\beta}}, \beta\in (0,1/2).$
We conclude that the probability that Receiver 1 decodes the information bits correctly approaches $1$ as $N$ goes to infinity. 

\vspace*{.1in}

{\em Security:} We will show that condition \eqref{security} is fulfilled for the sequence of $m$ blocks of transmission.
Note that Receiver 2 observes not only a realization of $Z^m$, but also estimates the RVs $U^m=(U^N(1),\dots,U^N(m))$ through procedure
\eqref{eq:r2-1}-\eqref{eq:r2-2}. For this reason the strong security condition to be proved takes the form
\begin{equation}
\lim_{N\to\infty} I(M^{m};\, U^{m}, Z^m)=0. \label{secrecy1}
\end{equation}
The proof of \eqref{secrecy1} is very similar to the proof of strong secrecy in Sec. \ref{sect:wt}.
The counterparts of Lemmas \ref{security_lemma_wiretap} and \ref{security_lemma_wiretap2}
for the BCC are provided below. The proofs are the same as the ones in Sec. \ref{sect:wt} and 
will be omitted.

\begin{lemma}
Let $T[{\mathscr A}]=\{t_i, i\in{\mathscr A}\}$ and $T[\cI\backslash {\mathscr A}]=\{t_i, i\in\cI\backslash{\mathscr A}\},$ 
where ${\mathscr A}$ is any subset of $\cI$.
Then
\begin{equation}
I(T[\cI\backslash{\mathscr A}]; T[{\mathscr A}], U^N, Z^N)=O(N^3 2^{-N^{\beta}}) \label{lemma_ineq_BCC}
\end{equation}
\label{security_lemma_BCC}
\end{lemma}
\begin{lemma}
Let   $M^j\triangleq(M_1,M_2,\dots,M_j)$, $Z^j\triangleq(Z_1,Z_2,\dots,Z_j),$
and $U^j\triangleq(U_1,U_2,\dots,U_j).$ Then,
for all $j=1,\dots, m$, we have
\begin{equation}
I(E_j; M^j, U^j, Z^j)=O(jN^3 2^{-N^{\beta}})
\end{equation}
\label{security_lemma_BCC2}
\end{lemma}
The rest of the proof also follows similarly to the proof in Sec. \ref{sect:wt}.
Repeating the inequalities which led to \eqref{wiretap_secure_7}, we get
\begin{align*}
&I(M^m;\, U^m, Z^m)  \leq I(M^m;\, U^m ,Z^m ,E_m) \\
& \leq I(M^{m-1}; U^{m-1}, Z^{m-1}, E_{m-1})+O(mN^3 2^{-N^{\beta}})
\end{align*}
which implies
\begin{align*}
I(M^m;\, U^m, Z^m) &\leq I(M^m;\, U^m, Z^m, E_m)\\&= O(m^2 N^3 2^{-N^{\beta}}).
\end{align*}
This completes the proof of \eqref{secrecy1}.
\end{proof}

Let us show that the ``additional-message'' rate $R_1$ as given by \eqref{eq:r} is also achievable. We have seen that
any rate pair $(R_0,R_s)$ satisfying \eqref{eq:r0}-\eqref{eq:rs} is achievable in the system. 
Moreover, observe that Receiver 1 decodes correctly messages at the rate of $\min [I(U;Y),I(U;Z)]$ according to
\eqref{eq:r1-1}-\eqref{eq:r1-2}, and additionally decodes messages at the rate of $I(V;Y|U)$ owing to the part of the encoding $\{t_i, i\in \cJ^{(1)}\}$ given by \eqref{eq:J1}. Since these two groups of information bits can be decoded simultaneously by Receiver 1, we conclude that it is possible 
to communicate to Receiver 1 an additional message at rate $R_1$.

Finally, we remind the reader that generalization of the results presented in
Sec. \ref{sect:wt} and Sec. \ref{sect:bcc} to nonbinary alphabets is possible using a multitude of methods
available in the literature. For example, E.~{\c S}a{\c s}o{\u g}lu \cite[Ch.4]{sasoglu12} 
showed that polarization phenomenon extends to all finite alphabets, proving also that the rate of
polarization is $2^{-N^{\beta}}$ for $\beta$ arbitrarily close to 1/2.
{For the set of subchannels having capacity close to 0, a similar result on the rate
of convergence holds \cite{karzand}.}

\section{Conclusion}
\label{conclusion}
In this paper, we have considered the wiretap channel problem \cite{wyner},\cite{csiszar2}. We proved
 that the secrecy capacity $C_s$ is
achievable by polar codes under the strong security constraint, and without
additional assumptions on the channel such as symmetry or degreadedness.
%The solution we propose for the wiretap channel problem is different from the
%previous works in the sense that it does not have any assumptions
%on the wiretap channel and achieves the capacity of the channel.
We also showed that it is possible to build on this solution by adding a second layer of encoding
which enables one to attain the capacity region of the BCC introduced by Csisz{\'a}r and K{\" o}rner.
 
Our construction and the 
construction of \cite{chou14} share some common features. In particular, the two-layer chaining construction for the 
BCC in our paper is equivalent to the chaining construction appearing in \cite{chou14}. In addition to that, \cite{chou14} 
includes
a third layer of encoding to handle the channel prefixing issue, i.e.,
the transition between the auxiliary RV $V$ and the real input $X$. Their construction also requires
the use of pre-shared common randomness (secret seed) which has vanishing rate compared to the overall
amount of communication. At the same time, \cite{chou14} bypasses the existence issue of good deterministic maps
of the type discussed in Sec.~\ref{sec:prelim}B by allowing additional secure communication between the users $X$ and $Y.$ 

%Finally, regarding the work \cite{wei14}, the chaining construction
%employed there is similar to the
%chaining idea we consider, however the coordinate partition in \cite{wei14} is
%suboptimal, which makes it difficult to achieve the strong security claim.

\ifCLASSOPTIONcaptionsoff
  \newpage
\fi

%\vspace{-0.29in}
\begin{IEEEbiographynophoto}{Talha Cihad Gulcu}
(S'13) received the B.S. degree from Middle East Technical University, Ankara, in 2009, the M.S. degree from Bilkent University, Ankara, in 2011, and the Ph.D. degree from University of Maryland, College Park, MD, in 2015, all in electrical engineering. His research interests are information theory and signal processing for communications.
\end{IEEEbiographynophoto}

%\vspace{5in}
\begin{IEEEbiographynophoto}{Alexander Barg}
(M'00-SM'01-F'08) received the M.Sc. degree in applied
mathematics and the Ph.D. degree in electrical engineering, the latter from
the Institute for Information Transmission Problems (IPPI) Moscow, Russia,
in 1987. He has been a Senior Researcher at the IPPI since 1988.
 Since 2003 he has been a Professor
in the Department of Electrical and Computer Engineering and Institute for
Systems Research, University of Maryland, College Park.

Alexander Barg was a co-recipient of the IEEE Information Theory Society
Paper Award in 2015. During 1997-2000, A. Barg was an Associate Editor
for Coding Theory of the IEEE Transactions on Information Theory.
He was the Technical Program Co-Chair of the 2006 IEEE International
Symposium on Information Theory and of 2010 and 2015 IEEE ITWs.
He serves on the Editorial Board of several journals including {\em Problems
of Information Transmission}, {\em SIAM Journal on Discrete Mathematics}, and
{\em Advances in Mathematics of Communications}.

Alexander Barg's research interests are in coding and information theory,
signal processing, and algebraic combinatorics.
\end{IEEEbiographynophoto}
\vfill

\end{document}